\documentclass[sigconf]{aamas}

\usepackage{balance} % for balancing columns on the final page
\usepackage{braket}
\usepackage{amsmath}
\usepackage{amsthm}
\newtheorem{assumption}{Assumption}
\newtheorem{problem}{Problem}
\usepackage{algorithm}
\usepackage{algpseudocode}
\usepackage{multirow}
\usepackage{verbatim}

\newcommand{\dist}{\mathrm{d}_{\mathcal{T}}} 
\newcommand{\normt}[1]{\|#1\|_{\mathcal{T}}} 
\newcommand{\inp}[2]{\langle #1, #2 \rangle} 
\newcommand{\inpt}[2]{\langle #1, #2 \rangle_{\mathcal{T}}} 
\newcommand{\logmap}[2]{\log_{#1}(#2)} 
\newcommand{\proj}{\Pi_{\mathcal{T}}}

\def\mathbi#1{\textbf{\em #1}}

\newtheorem{theorem}{Theorem}

\doi{ADPL7324}

\makeatletter
\gdef\@copyrightpermission{
  \begin{minipage}{0.2\columnwidth}
   \href{https://creativecommons.org/licenses/by/4.0/}{\includegraphics[width=0.90\textwidth]{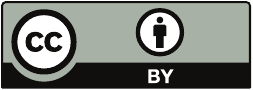}}
  \end{minipage}\hfill
  \begin{minipage}{0.8\columnwidth}
   \href{https://creativecommons.org/licenses/by/4.0/}{This work is licensed under a Creative Commons Attribution International 4.0 License.}
  \end{minipage}
  \vspace{5pt}
}
\makeatother

\setcopyright{ifaamas}
\acmConference[AAMAS '26]{Proc.\@ of the 25th Intern. Conf.
on Autonomous Agents and Multiagent Systems (AAMAS 2026)}{May 25--29, 2026}
{Paphos, Cyprus}{C.~Amato, L.~Dennis, V.~Mascardi, J.~Thangarajah (eds.)}
\copyrightyear{2026}
\acmYear{2026}
\acmDOI{}
\acmPrice{}
\acmISBN{}

\title[Distributed Quantum Gaussian Processes]{Distributed Quantum Gaussian Processes for Multi-Agent Systems}

\author{Meet Gandhi}
\affiliation{
  \institution{Colorado School of Mines}
  \city{Golden}
  \country{USA}}
\email{meet_gandhi@mines.edu}

\author{George P. Kontoudis}
\affiliation{
  \institution{Colorado School of Mines}
  \city{Golden}
  \country{USA}}
\email{george.kontoudis@mines.edu}

\begin{abstract}

Gaussian Processes (GPs) are a powerful tool for probabilistic modeling, but their performance is often constrained in complex, large-scale real-world domains due to the limited expressivity of classical kernels. Quantum computing offers the potential to overcome this limitation by embedding data into exponentially large Hilbert spaces, capturing complex correlations that remain inaccessible to classical computing approaches. In this paper, we propose a Distributed Quantum Gaussian Process (DQGP) method in a multi-agent setting to enhance modeling capabilities and scalability. To address the challenging non-Euclidean optimization problem, we develop a Distributed consensus Riemannian Alternating Direction Method of Multipliers (DR-ADMM) algorithm that aggregates local agent models into a global model. We evaluate the efficacy of our method through numerical experiments conducted on a quantum simulator in classical hardware. We use real-world, non-stationary elevation datasets of NASA's Shuttle Radar Topography Mission and synthetic datasets generated by Quantum Gaussian Processes. Beyond modeling advantages, our framework highlights potential computational speedups that quantum hardware may provide, particularly in Gaussian processes and distributed optimization.

\end{abstract}


\newcommand{\BibTeX}{\rm B\kern-.05em{\sc i\kern-.025em b}\kern-.08em\TeX}

\begin{document}

%%% The following commands remove the headers in your paper. For final 
%%% papers, these will be inserted during the pagination process.

\pagestyle{fancy}
\fancyhead{}

%%% The next command prints the information defined in the preamble.

\maketitle

\noindent\textbf{Code:} \href{https://github.com/mpala-lab/distributed-quantum-gaussian-processes}{github.com/mpala-lab/distributed-quantum-gaussian-processes}

%%%%%%%%%%%%%%%%%%%%%%%%%%%%%%%%%%%%%%%%%%%%%%%%%%%%%%%%%%%%%%%%%%%%%%%%

\section{Introduction}

Decision-making in autonomous systems relies on reliable %risk
uncertainty quantification. Gaussian processes (GPs), as an inherently probabilistic modeling technique, satisfy %this 
the need through accurate predictions and principled uncertainty estimation. To learn a GP model that characterizes %of an unknown space, an automated agent interacts with its environment to sample \textit{interesting} data points.
the intrinsic dynamics of an unknown environment, an agent typically samples informative data points from the environment. 
%With a single agent, this exploration process becomes cumbersome. 
However, training a GP model on $N$ samples involves~$\mathcal{O}(N^3)$ computations and $\mathcal{O}(N^2)$ memory. When a single agent is responsible for exploration and computation, the process becomes not only computationally demanding but also time-sensitive, as the agent must physically traverse all locations to gather samples. 
%with $N$ sampled data points in the space, a GP model training typically demands $\mathcal{O}(N^3)$ computational time complexity and $\mathcal{O}(N^2)$ space complexity. 
This limitation restricts the applicability of GPs to large-scale datasets and %expansive 
environments---conditions commonly encountered in %realistic 
autonomous systems---especially for single-agent systems. To this end, GP approximations have been introduced that can be broadly categorized into two main classes: exact aggregation methods and inducing point-based approximation methods~\cite{liu2020gaussian}. We focus on the former class which serve as distributed GP (DGP) approaches~\cite{deisenroth2015distributed,kontoudis2021decentralized,kontoudis2025multi}. %have been developed, %facilitating the training of GP models for data sizes that are infeasible with a single agent. 
These methods enable GP training on dataset sizes that would otherwise be infeasible for a single agent. In particular, they deploy multiple agents % or computational nodes 
in local %ly varied 
regions of the input space, %enabling the learning of distinct local GP models relevant to the unknown space. 
allowing each to learn a local GP model that captures regional characteristics. 
%These 
The local GP models %eventually contribute to forming a global GP model for the environment through multi-agent collaboration, which can occur in either a centralized \cite{xie2019distributed} or decentralized manner~\cite{kontoudis2024scalable}. 
are then aggregated through multi-agent coordination to form a global GP model. %, either in centralized~\cite{xie2019distributed} or decentralized~\cite{kontoudis2024scalable} settings.
%Note that in DGP, the dataset is stored locally, and the overall computation burden is also divided among the agents, making DGP applicable to extensive unknown spaces. Therefore, DGP solves the scalability issue with GP.
 FACT-GP \cite{deisenroth2015distributed} and its generalized version g-FACT-GP \cite{liu2018generalized} %belong to the former class, where the 
 enforce partitioning of sampled data, % is partitioned, 
 and the resulting local posteriors are subsequently aggregated. In addition, apx-GP \cite{xie2019distributed} and gapx-GP~\cite{kontoudis2023decentralized} reach global consensus by using the multi-agent Alternating Direction Method of Multipliers (ADMM)~\cite{chang2014multi}. By distributing both data storage and computational effort among agents, DGP methods effectively overcome the scalability limitations of standard GPs. %, making probabilistic learning in complex, large-scale domains feasible.

Gaussian Processes employ kernel functions \cite{kanagawa2018gaussian} to model the correlations among the data points by projecting them into a high-dimensional feature space. This %kernel 
mapping enables GPs to capture complex relationships. However, %these 
the classical kernels % functions 
possess limited expressivity due to the underlying mathematical formulation that is tractable on classical hardware. This shortcoming can be addressed through the emerging field of quantum computing. Our goal in this work is to leverage quantum computing to develop powerful and scalable GPs. Specifically, we aim to design a distributed framework for Quantum Gaussian Processes %(DQGPs) 
that exploits the expressive capability of quantum kernels while efficiently distributing the computational and memory load across multiple agents. % in large-scale complex systems.} %{\color{red} a high level objective of the proposed method  is missing}

%{\color{red}You need a paragraph here to discuss distributed Gaussian process training methods, e.g., FACT-GP~\cite{deisenroth2015distributed}, g-FACT-GP~\cite{liu2018generalized}, apx-GP~\cite{xie2019distributed}, gapx-GP~\cite{kontoudis2024scalable}. }

The current generation of quantum hardware, termed as the NISQ (Noisy Intermediate-Scale Quantum) era~\cite{preskill2018quantum}, lacks %the property of 
fault tolerance, making it %harder to achieve 
challenging to realize a clear quantum advantage. This %has given rise to 
limitation has motivated the development of several hybrid quantum-classical techniques, termed as variational quantum algorithms (VQAs)~\cite{cerezo2021variational}. VQAs employ parameterized quantum circuits~\cite{benedetti2019parameterized}, where each circuit parameter serves as an optimization variable adjusted to %achieve a predefined goal for a system. %In VQA, s
minimize a cost function. The system dynamics are modeled within the quantum domain, %leading to much
enabling faster gradient %computation 
evaluation~\cite{wierichs2022general}, %after which optimizing parameters are updated using classical algorithms. 
while parameter updates are performed using classical optimizers. A notable example of VQA is the Quantum Approximate Optimization Algorithm (QAOA)~\cite{guerreschi2019qaoa}, % and has been used 
widely applied to %to solve 
combinatorial optimization problems. %Note that with 
However, VQAs face several challenges, including %A disadvantage of VQAs is that 
the optimization landscape that often contains %gets stuck in vast, 
large, flat regions %, leading to the vanishing gradient problem. 
which cause gradients to vanish. This phenomenon is called barren plateaus \cite{larocca2025barren} %and is an active field of research
and remains an active research topic~\cite{mcclean2018barren, patti2021entanglement, sack2022avoiding, peng2025breaking, cunningham2025investigating}. Moreover, with the advancement of quantum hardware, exemplified by the milestone achievement of quantum supremacy \cite{arute2019quantum}, it becomes crucial to devise quantum algorithms that can be implemented in a parallelized and distributed fashion \cite{parekh2021quantum}, where quantum circuit evaluations can be allocated %across 
to multiple quantum processors to enhance scalability and robustness. %In recent years, quantum computing has been applied to classical control problems~\cite{berberich2024quantum} projecting classical data into quantum states, which lie in a classically intractable quantum Hilbert space. This exponentially large space enables the representation of subtle and complicated covariances that are inaccessible to classical kernels, hence allowing for the creation of highly expressive GP models. 

\begin{figure*}[!t]
	\includegraphics[width=.985\textwidth]{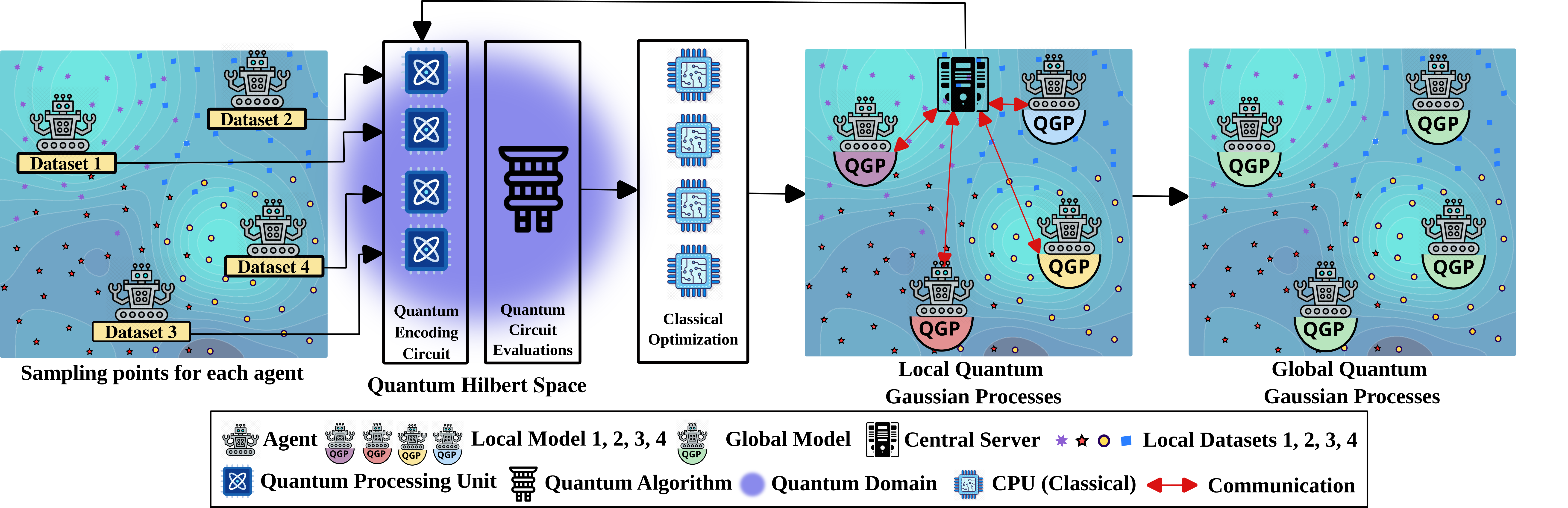}
	\centering
	\caption{Distributed Quantum Gaussian Process (DQGP): A hybrid classical-quantum framework for multi-agent systems.}
	\label{fig:application}
\end{figure*}

A major advancement in quantum computing in recent years has been the introduction of quantum kernel functions~\cite{huang2021power}, %resulting in their use in 
which have led to the development of Quantum Gaussian Processes (QGPs) \cite{rapp2024quantum}. % among others~\cite{arceci2024gaussian, smith2023faster}. %{\color{red}examples}. 
Quantum kernels incorporate parameterized quantum circuits, called Quantum Encoding Circuits \cite{hubregtsen2022training}, to embed %the 
classical data into the quantum domain. Considerable research work %is dedicated 
has been conducted to %devising 
designing %these novel quantum kernel functions 
kernels %aiming ideal target 
with the goal of achieving optimal alignment between quantum feature space and %with the 
classical data labels. %In addition to 
Beyond QGPs, these quantum kernels have proven beneficial in other domains of %quantum 
machine learning~\cite{arceci2024gaussian, smith2023faster}. % tasks~\cite{arceci2024gaussian, smith2023faster}. 
In particular,%~\cite{yun2022quantum, park2023quantum, chen2025quantum} 
~\cite{park2023quantum, chen2025quantum} employ quantum kernels with reinforcement learning in multi-agent settings.
%The exploration of these quantum kernel-based algorithms is facilitated by a prominent Python library $sQUlearn$ \cite{kreplin2025squlearn}. In our work, we utilize the encoding circuits and kernels from $sQUlearn$. 
%These algorithms are commonly implemented using libraries such as \textsc{sQUlearn}~\cite{kreplin2025squlearn}, a Python framework for quantum kernel-based learning. In our work, we adopt quantum encoding circuits and kernels provided by \textsc{sQUlearn}. 
%Apart from 
In addition to QGP training, significant efforts have been %made 
directed to leveraging quantum algorithms for GP inference. Numerous methods~\cite{zhao2019quantum,kus2021sparse,chen2022quantum,farooq2024quantum,galvis2025provable} have been formulated based on the Harrow-Hassidim-Lloyd (HHL) algorithm \cite{lloyd2010quantum} %, which solves the quantum system of linear equations. These methodologies \cite{chen2022quantum}, \cite{zhao2019quantum}, \cite{farooq2024quantum}, \cite{galvis2025provable} and \cite{kus2021sparse} develop quantum circuits intending to speed up the computation of inverse kernel matrices required for GP inference. 
%have been proposed to 
%and 
to accelerate the computation of inverse kernel matrices required for prediction.

% A %profound 
% theoretical %result was presented recently, where 
% connection between %specialized quantum circuits called 
% Quantum Neural Networks (QNNs) and %Gaussian Processes 
% GPs is established in~\cite{garcia2025quantum,girardi2025trained}. %It mentions that u
% Under specific conditions, a QNN with high expressive powers, i.e., %with an 
% increasing number of qubits or layers, has been shown to converge to a GP. The theoretical equivalence provides a powerful foundation for studying the GPs through the lens of QNNs and vise versa. %, %therefore, reinforcing the need for analyzing and understanding the fundamentals of QGPs further. 
% %underscoring the importance of understanding the fundamental behavior of QGPs.
% %To efficiently investigate QGPs with large-scale datasets and real-world environments, the scalability and operation in a multi-agent active learning setting is essential. 
% In this work, we %provide a first study into 
% introduce a Distributed Quantum Gaussian Process (DQGP) framework to extend the modeling and computational capabilities of QGPs. The concept overview of our proposed approach is shown in Fig.~\ref{fig:application}.

\subsubsection*{Contribution} The contribution of this paper is twofold. First, we introduce the distributed consensus Riemannian ADMM (DR-ADMM) %algorithm, a distributed 
optimization %scheme 
that can efficiently train parameterized quantum circuits across multiple agents. Next, we formulate the distributed quantum gaussian process (DQGP), % framework, 
which successfully scales the expressive power of quantum kernels utilizing quantum aspects of \textit{entanglement}, \textit{superposition} and \textit{measurements} to handle complex, multi-agent scenarios (Fig.~\ref{fig:application}). %Rigorous 
Numerical experiments on non-stationary fields %establish our proposed methodology's 
demonstrate the %superior 
enhanced performance of the proposed method %when 
compared to other classical %distributed GP~
approaches.
%In Section \ref{sec2}, we first introduce the mathematical formulations of classical GPs, then establish the quantum computing principles necessary for constructing quantum GPs, and finally lay out the classical distributed GP approach. Thereafter, we present our two contributions, Distributed Consensus Riemannian ADMM and Distributed Quantum Gaussian Process, in Section \ref{sec3}. Ultimately, in Section \ref{sec4}, we put forth our numerical experiments setup and the results.

%%%%%%%%%%%%%%%%%%%%%%%%%%%%%%%%%%%%%%%%%%%%%%%%%%%%%%%%%%%%%%%%%%%%%%%%

\section{%Foundations and 
Problem Formulation}\label{sec2} % this is bad labeling
In this section, we discuss classical GPs, the function-space view of GPs that connects classical to quantum computing, quantum GP regression, distributed classical GP training, and state the problem. % of distributed quantum GP training. 

% \subsection{Preliminaries}
% The notation is standard throughout the paper. $\mathbb{R}_{>0}$ and $\mathbb{Z}_{>0}$ represent the set of positive real numbers and integers, respectively. Similarly, $\mathbb{R}_{\geq0}$ and $\mathbb{Z}_{\geq0}$ denote the set of non-negative real numbers and integers. A variable in vector, matrix, and scalar form is depicted by $\mathbi{x}$, $\mathbi{X}$, and $x$ respectively. The superscript in parentheses $\mathbi{x}^{(s)}$ symbolizes the vector $\mathbi{x}$ in the $s^{\text{th}}$ iteration of a recursive algorithm. The subscript $x_i$ represents the $i^{\text{th}}$ element of vector $\mathbi{x}$. Hence, a set of the elements of vector $\mathbi{x} \in \mathbb{R}^N$ can be denoted as $\{x_i\}^{N}_{i=1}$. The subscript of a vector $\mathbi{y}_m$ depicts the vector $\mathbi{y}$ for $m^{\text{th}}$ agent. The absolute value of a scalar, a vector, or a complex number is symbolized as $\vert \,.\, \vert$, the $L_2$ norm as $||.||_2$ and the infinity norm as $||.||_{\infty}$. The cardinality of a set $V$ is denoted by $\text{card}(V)$. The computational complexities are depicted by $\mathcal{O}(.)$. Note that a quantum kernel $\kappa$ here can take the form of a function or a matrix; therefore, $\kappa$ represents the functional form, whereas $\boldsymbol{\kappa}$ represents the matrix form. 

\subsection{Classical Gaussian Processes (GPs)}
GP regression is a non-parametric Bayesian modeling approach \cite{ghahramani2013bayesian} that %aims at learning 
provides a probability distribution over an infinite-dimensional space of functions.  The system observations are modeled as %of the form 
$y(\textbf{\textit{x}}) = f(\textbf{\textit{x}}) + \epsilon$, where $\textbf{\textit{x}} \in \mathbb{R}^D$ is the %data points in 
input of a \textit{D} dimensional space, $y \in \mathbb{R}$ is the corresponding label, $f(\textbf{\textit{x}}) \sim \mathcal{GP} (0, k(\textbf{\textit{x}}, \textbf{\textit{x}}'))$ is a zero-mean GP %generator function 
with covariance function $k: \mathbb{R}^D \times \mathbb{R}^D \rightarrow \mathbb{R}$, % encoding the correlation between data points $\textbf{\textit{x}}$ and $\textbf{\textit{x}}'$ 
and $\epsilon \sim \mathcal{N} (0, \sigma^2_\epsilon)$ the i.i.d. zero mean Gaussian measurement noise with variance $\sigma^2_{\epsilon} > 0$. The objective of GP regression then is to estimate the latent function $f$ given dataset $\mathcal{D} = \{\mathbi{X}, \mathbi{y}\}$ with inputs $\mathbi{X} = \{\mathbi{x}_{n}\}^{N}_{n=1}$, outputs $\mathbi{y} = \{y_n\}^{N}_{n=1}$, where $N$ is the number of observations. 

%\subsubsection{Gaussian Process Training}
 %This is accomplished by determining 
 We determine the hyperparameters $\boldsymbol{\theta}$ of the covariance function by using maximum likelihood estimation. % that %of the covariance function that maximize 
 %employs t
 The negative marginal log-likelihood function takes the form of, %which is given by,
\begin{eqnarray}\label{eqn:LL}
    \mathcal{L} = \log p(\mathbi{y}\,\vert \,\mathbi{X}) =  \frac{1}{2} \Big( \mathbi{y}^\intercal \mathbi{C}^{-1}_{\boldsymbol{\theta}} \mathbi{y} + \log \vert\mathbi{C}_{\boldsymbol{\theta}} \vert + N \log(2\pi) \Big),
\end{eqnarray}
where $\mathbi{C}_{\theta} = \mathbi{K} \, + \, \sigma^{2}_{\epsilon}\, \mathbi{I}_{N}$ is the positive definite covariance matrix %needed for numerical stability with $\boldsymbol{\theta}$ as the hyperparameter set 
and $\mathbi{K} = k(\mathbi{X}\, , \mathbi{X}\,) \succeq 0 \in \mathbb{R}^{N \times N}$ is the positive semi-definite correlation matrix between %data points 
inputs of \mathbi{X}. The GP training %regression 
problem %ultimately can be 
yields, %is formulated as, % the minimization of the negative marginal log-likelihood (NLL) function  as follows,
\begin{eqnarray}\label{eqn:gpopt}
    \hat{\boldsymbol{\theta}} = \underset{\boldsymbol{\theta}}{\arg\min}  \, \mathbi{y}^\intercal \mathbi{C}^{-1}_{\boldsymbol{\theta}} \mathbi{y} + \log \vert\mathbi{C}_{\boldsymbol{\theta}} \vert,
\end{eqnarray}
which can be solved using gradient-based optimization methods %require the %computation of 
with partial derivative of the objective, %that takes the form of, %NLL partial derivative w.r.t. the hyperparameter set $\boldsymbol{\theta},$
\begin{eqnarray}\label{eqn:gradient}
    \frac{\partial \mathcal{L}(\boldsymbol{\theta})}{\partial \boldsymbol{\theta}} = \frac{1}{2} \mathrm{Tr} \left\{ \left(\mathbi{C}^{-1}_{\boldsymbol{\theta}} - \mathbi{C}^{-1}_{\boldsymbol{\theta}} \mathbi{y} \mathbi{y}^\intercal \mathbi{C}^{-1}_{\boldsymbol{\theta}} \right) \frac{\partial \mathbi{C}_{\boldsymbol{\theta}}}{\partial \boldsymbol{\theta}} \right\}.
\end{eqnarray}

% {\color{green}here}

\subsection{Function-Space View of Gaussian Processes}
The preceding discussion of GPs focuses on a formulation in a finite-dimensional space associated with a discrete dataset $\{\mathbi{X}, \mathbi{y}\}$. We can generalize to an infinite-dimensional function-space view, where a GP is defined as $ f(\mathbi{x}) \sim \mathcal{GP} (m(\mathbi{x}\,), k(\textbf{\textit{x}}, \textbf{\textit{x}}'))$ with $m(\mathbi{x}\,)$ the mean function and $k(\textbf{\textit{x}}, \textbf{\textit{x}}')$ the covariance function~ \cite{shi2011gaussian}. %Now, GP can be defined as a collection of random variables that perform the task of assigning probability to every possible function based on the prior derived from the covariance function.
In this way, a GP can be interpreted as a collection of random variables, any finite subset of which follows a joint Gaussian distribution. This allows the GP to assign probabilities over the space of possible functions, governed by the prior induced by the covariance. The function-space view provides a critical link to the kernel function~$\kappa$ %with feature map $\phi$ written as, 
through the associated feature map $\phi$, expressed as,
\begin{eqnarray} \label{eqn:genkernel}
    \kappa(\mathbi{x}, \mathbi{x}') = \phi (\mathbi{x})^\intercal \phi(\mathbi{x}').
\end{eqnarray}
In classical computing settings, the kernel function measures the correlation between two data points in a high-dimensional feature space without explicitly mapping the data into that space. In the quantum domain, the classical data are encoded into quantum states within a \textit{quantum Hilbert space} using \textit{Quantum Encoding Circuits}~\cite{schuld2019quantum}. The Hilbert space is exponentially large, and thus a classical computer %ing machine
would struggle even to represent the quantum states, let alone compute their inner product~\eqref{eqn:genkernel}. \textit{Quantum kernel functions}~\cite{schuld2021quantum} enable the estimation of correlations between these encoded quantum states, capturing complex and highly non-linear relationships in the original classical data. %can then be utilized to estimate the correlation between the mapped quantum states. This would allow for modeling more complex and non-linear relationships between the original classical data.

%\subsubsection{\textbf{Quantum Encoding Circuits}.\\}
\subsection{Quantum Gaussian Processes (QGPs)}
The first step in constructing a Quantum Gaussian process (QGP) is to encode classical data into quantum states. This is achieved using Quantum encoding circuits which %Quantum encoding circuits can be expressed as 
implement a mapping from a classical data vector $\mathbi{x} \in \mathbb{R}^D$ to a quantum state $\ket{\psi_{\mathbi{x}}}$ in quantum Hilbert space $\mathcal{H},$ given by $\Phi : \mathbi{x} \rightarrow \ket{\psi_{\mathbi{x}}} = U(\mathbi{x}, \boldsymbol{\theta}) \ket{0}^{\otimes q},$
% \begin{eqnarray*}
%     \Phi : \mathbi{x} \rightarrow \ket{\psi_{\mathbi{x}}} = U(\mathbi{x}, \boldsymbol{\theta}) \ket{0}^{\otimes q},
% \end{eqnarray*}
where $\Phi$ is the encoding map, $U(\mathbi{x}, \boldsymbol{\theta})$ a unitary operator representing the entire quantum encoding circuit with $\boldsymbol{\theta}$ the trainable hyperparameters, and $\ket{0}^{\otimes q}$ the initial quantum state of a system comprising $q$ qubits. The circuit $U(\mathbi{x}, \boldsymbol{\theta})$ is composed of quantum unitary logic gates~\cite{Williams2011}, primarily rotational gates $\mathbi{R}_{\mathbi{X}}(\boldsymbol{\sigma}_{\mathbf{X}}, \theta_{x}), \mathbi{R}_{\mathbi{Y}}(\boldsymbol{\sigma}_{\mathbf{Y}},\theta_{y}), \mathbi{R}_{\mathbi{Z}}(\boldsymbol{\sigma}_{\mathbf{Z}},\theta_{z})$ with $2 \times 2$ Pauli matrices $(\boldsymbol{\sigma}_{\mathbf{X}}, \boldsymbol{\sigma}_{\mathbf{Y}}, \boldsymbol{\sigma}_{\mathbf{Z}})$ and their controlled versions. Additional gates include the Hadamard $\mathbf{H}$, Phase $\mathbf{P}$, $\mathbf{CNOT}$, and %occasionally 
the $\mathbf{SWAP}$ gate. %Their matrix forms yields,
% \begin{align}\label{eqn:quantumgates}
% \mathbi{R}_{\mathbi{X}}(\theta_{x}) = \exp{\bigg\{-i\theta_{x}\begin{pmatrix} 0 & 1 \\ 1 & 0\end{pmatrix}\bigg\}} &= \begin{bmatrix} \cos(\theta_x/2) & -i\sin(\theta_x/2) \\ -i\sin(\theta_x/2) & \cos(\theta_x/2) \end{bmatrix}, \nonumber
% \\
% \mathbi{R}_{\mathbi{Y}}(\theta_{y}) = \exp{\bigg\{-i\theta_{y}\begin{pmatrix} 0 & -i \\ i & 0\end{pmatrix}\bigg\}} &= \begin{bmatrix} \cos(\theta_y/2) & -\sin(\theta_y/2) \\ \sin(\theta_y/2) & \cos(\theta_y/2)  \end{bmatrix}, \nonumber
% \\
% \mathbi{R}_{\mathbi{Z}}(\theta_{z}) = \exp{\bigg\{-i\theta_{z}\begin{pmatrix} 1 & 0 \\ 0 & -1\end{pmatrix}\bigg\}} &= \begin{bmatrix} e^{-i\theta_z/2} & 0 \\ 0 & e^{i\theta_z/2}   \end{bmatrix}, \nonumber
% \\
% \mathbf{CNOT} = \begin{bmatrix} 1&0&0&0 \\ 0&1&0&0 \\ 0&0&0&1 \\ 0&0&1&0\end{bmatrix}&, \quad
% \mathbf{SWAP} = \begin{bmatrix} 1&0&0&0 \\ 0&0&1&0 \\ 0&1&0&0 \\ 0&0&0&1\end{bmatrix}, \nonumber
% \\
% \mathbf{H} = \frac{1}{\sqrt{2}}\begin{bmatrix} 1 & 1 \\ 1 & -1\end{bmatrix}&, \quad
% \mathbf{P} = \begin{bmatrix} 1 & 0 \\ 0 & i\end{bmatrix}.
% \end{align}
%If the encoding circuit contains solely of rotational gates, then the trainable hyperparameter vector $\boldsymbol{\theta}$ is purely periodic. 
Moreover, %the circuit 
$U(\mathbi{x}, \boldsymbol{\theta})$ can be $\iota$-layered with an identical gate structure in each %of the $\iota$ 
layer, i.e., 
%\begin{eqnarray*}
$U(\mathbi{x}, \boldsymbol{\theta}) = U_{\textrm{final}} U_{\iota}(\mathbi{x}, \boldsymbol{\theta}_\iota) \dots U_{1}(\mathbi{x}, \boldsymbol{\theta_1})$.
%\end{eqnarray*}

% {\color{red}Is it necessary to include the matrix form of all these gates? Also, do we use the concepts described in the next two paragraphs? I am removing them from the conference paper to save some space, but if they are important we can add them back.}

% A significant body of research focuses on designing effective quantum encoders with varied gate compositions %that can be utilized in 
% suitable for hybrid quantum-classical systems. One such approach %is finding 
% involves optimizing $\boldsymbol{\theta}$ %that maximizes 
% to maximize the target alignment objective function~\cite{coelho2025quantum} for a dataset $\mathcal{D} = \{\mathbi{X}, \mathbi{y}\}$,
% \begin{eqnarray*}
%     \mathcal{TA}(\Phi_{\boldsymbol{\theta}}) = \frac{\sum_{i,j} y_i y_j \Phi_{\boldsymbol{\theta}}(\mathbi{x}_i, \mathbi{x}_j)}{\sqrt{\sum_{i,j} (\Phi_{\boldsymbol{\theta}}(\mathbi{x}_i, \mathbi{x}_j))^2 \sum_{i,j} y_i^2 y_j^2}}.
% \end{eqnarray*}

% \subsubsection{\textbf{Quantum Kernel Functions}.\\}

After encoding the classical data vectors $\mathbi{x}$ and $\mathbi{x}'$ into the quantum states $\ket{\psi_\mathbi{x}}$ and $\ket{\psi_{\mathbi{x}'}}$ respectively, the quantum kernel function~$\kappa$ computes the correlation between the states which can then be used to populate the GP covariance matrix $\mathbi{C}_{\boldsymbol{\theta}} (\mathbi{x}, \mathbi{x}')$. % of a GP. 
A %notable example of the 
representative quantum kernel is the fidelity kernel $\kappa_{\mathcal{F}}$ \cite{zanardi2006ground} derived from the fidelity measure $\mathcal{F}$, % that follows, % commonly used in quantum computing. It is expressed as,
\begin{eqnarray}\label{eqn:fidelityker}
    \kappa_{\mathcal{F}}: \mathcal{H} \times \mathcal{H} \rightarrow [0,1] = \vert\braket{\psi_{\mathbi{x}}|\psi_{\mathbi{x}'}}\vert^2,
\end{eqnarray}
where $\kappa_{\mathcal{F}}=1$ %value of 1 symbolizes 
shows complete overlap between two quantum states %whereas 0 
and $\kappa_{\mathcal{F}}=0$ indicates that the quantum states are orthogonal. There exists mathematical relevance between \eqref{eqn:genkernel},~\eqref{eqn:fidelityker}, and the quantum kernels serving as covariance functions in QGP training. %It is worth noting the mathematical relevance between \eqref{eqn:genkernel} and~\eqref{eqn:fidelityker}. %These 
%{\color{red} The} quantum kernels serve as covariance functions in QGP training.

% \subsubsection{\textbf{Quantum Gaussian Process}.\\}
% Quantum kernels are utilized as covariance functions for training a Quantum Gaussian Process. The increase in expressivity of quantum kernels owing to the classically intractable quantum Hilbert space allows for creating models with superior predictive performance. However, increasing the size of the training data set adversely affects the performance of Quantum GP.

\subsection{Distributed Gaussian Processes (DGPs)}
In classical GP training~\eqref{eqn:gpopt},
%GP is a powerful modeling methodology; however, its limitation lies in scalability. From Equation (\ref{eqn:gpopt}), 
each optimization round entails time complexity of %iteration of GP training 
$\mathcal{O}(N^3)$, due to the computation of covariance matrix inverse $\mathbi{C}^{-1}_{\boldsymbol{\theta}}$. %We also need to store 
Additionally, storing $\mathbi{C}^{-1}_{\boldsymbol{\theta}}$ and $N$ dataset size requires $\mathcal{O}(N^2 + DN)$ space complexity. %, where $D$ is the input dimensionality. 
The high computational and memory demands make classical %and Quantum 
GPs impractical for large-scale, real-world applications. DGP \cite{deisenroth2015distributed} %solves the large dataset scalability problem 
addresses the scalability bottleneck by distributing both computation and storage across multiple agents, under the assumption of local dataset independence.
\begin{assumption}\label{assumption:local}
    All local datasets represent distinct areas %, and hence 
    with local %ly learned agent models 
    models being statistically independent. 
\end{assumption}
Instead of computing the large covariance inverse $ \mathbi{C}^{-1}_{\boldsymbol{\theta}}$, a DGP approximation yields, $\mathbi{C}^{-1}_{\boldsymbol{\theta}} \approx \mathrm{diag}(\mathbi{C}^{-1}_{\boldsymbol{\theta}, 1}, \ldots, \mathbi{C}^{-1}_{\boldsymbol{\theta}, M})$,
% \begin{eqnarray*}
%     \begin{bmatrix} \mathbi{C}^{-1}_{\boldsymbol{\theta}} \end{bmatrix} \rightarrow 
% \begin{bmatrix}
% \mathbi{C}^{-1}_{\boldsymbol{\theta}_1, 1} & & & \\
% & \mathbi{C}^{-1}_{\boldsymbol{\theta}_2, 2} & & \\
% & & \ddots & \\
% & & & \mathbi{C}^{-1}_{\boldsymbol{\theta}_M, M}
% \end{bmatrix}
% \end{eqnarray*}
where $\mathbi{C}^{-1}_{\boldsymbol{\theta}} \in \mathbb{R}^{N \times N}$ and $\mathbi{C}^{-1}_{\boldsymbol{\theta}, m} \in \mathbb{R}^{N_m \times N_m}$ for all agents $m \in [1, M]$. % with M agents. % and $\boldsymbol{\theta}_m$ as the hyperparameter vector trained on agent $m$'s local dataset. 
Hence, by virtue of Assumption \ref{assumption:local}, DGPs can decompose the optimization problem over dataset $\mathcal{D}$ into a distributed optimization problem over local datasets $\mathcal{D}_m$.~
%are essentially converting an optimization process on the global dataset $\mathcal{D}$, giving $\boldsymbol{\theta}$ to several optimization processes on local datasets $\mathcal{D}_m$ outputting the local hyperparameter vectors $\boldsymbol{\theta}_m$ for $m \in [1, M]$.
%Centralized GP training techniques are formulated for distributed GP based on the alternating direction method of multipliers (ADMM). 
DGP training methods are also formulated using the multi-agent alternating direction method of multipliers (ADMM)~\cite{boyd2011admm,chang2014multi}. The analytical proximal GP (apx-GP) \cite{xie2019distributed} training employs a first-order approximation on the local log-likelihood function $\mathcal{L}_m$ under the assumption of Lipschitz continuity,
\begin{assumption}\label{assumption:L}
    The function $\mathcal{L}_m: \mathbb{R^N} \rightarrow \mathbb{R}$ is Lipschitz continuous with a positive parameter $L > 0$ if,
    \begin{eqnarray*}
        {{\Vert \nabla \mathcal{L}_m(\mathbi{x}) - \nabla \mathcal{L}_m(\mathbi{y}) \Vert}_2 \leq L {\Vert \mathbi{x} - \mathbi{y} \Vert}_2}, \hspace{1cm} \forall \mathbi{x}, \mathbi{y} \in \mathbb{R^N}.
    \end{eqnarray*}
\end{assumption}
The main idea behind apx-GP and other ADMM-based GP training algorithms \cite{kontoudis2024scalable} is that every agent $m$ is allowed to have an opinion on its hyperparameter vector $\boldsymbol{\theta}_m$; however, once the optimization is complete, they should agree on a global consensus vector $\mathbi{z}$. The optimization scheme of apx-GP %'s $s^{th}$ iteration involves the following steps,
yields,
\begin{subequations}
\begin{alignat}{3}
    \mathbi{z}^{(s+1)} &= \frac{1}{M} \sum^{M}_{m=1} \left( \boldsymbol{\theta}^{(s)}_{m} + \frac{1}{\rho}  \boldsymbol{\psi}^{(s)}_{m} \right) \label{eqn:admmz}\\
    \boldsymbol{\theta}^{(s+1)}_{m} &= \mathbi{z}^{(s+1)} - \frac{1}{\rho + L_{m}} \left( \nabla_{\boldsymbol{\theta}} \mathcal{L} (\mathbi{z}^{(s+1)}) + \boldsymbol{\psi}^{(s)}_{m} \right)\label{eqn:admmt}\\
    \boldsymbol{\psi}^{(s+1)}_{m} &= \boldsymbol{\psi}^{(s)}_{m} + \rho \left( \boldsymbol{\theta}^{(s+1)}_{m} - \mathbi{z}^{(s+1)} \right)\label{eqn:admmp}
\end{alignat}
\label{eq:apxGP}%
\end{subequations}
where $\rho > 0$ is the parameter promoting the consensus between all %the 
$\boldsymbol{\theta}_{m}$ and $\mathbi{z}$, $L_m$ is a positive Lipschitz constant for each agent $m$, $\nabla_{\boldsymbol{\theta}} \mathcal{L} (\mathbi{z}^{(s+1)})$ is the gradient of marginal log-likelihood function %with respect to the hyperparameter vector 
with respect to $\mathbi{z}^{(s+1)}$, and $\boldsymbol{\psi}_m$ is the dual variable vector for each agent $m$. The reduced time and space complexity of apx-GP is $\mathcal{O}(N_m)=\mathcal{O}(N^3/M^3)$ and $\mathcal{O}(N^2/M^2 + D(N/M))$ respectively.

\subsection{Problem Statement}
While DGPs effectively alleviate the scalability limitations of standard GPs, their predictive performance remains constrained by the limited expressivity of classical kernels. %is still restricted by the expressive power of classical kernels. 
Quantum kernels, on the contrary, can capture subtle correlations that are inaccessible to classical kernels by leveraging the exponentially large Hilbert space for data mapping. %Therefore, in this work, we 
Motivated by this, our work focuses on scaling QGPs to enable learning in multi-agent systems. %This requires addressing the following problems.

\begin{figure*}[!t]
	\includegraphics[width=\textwidth]{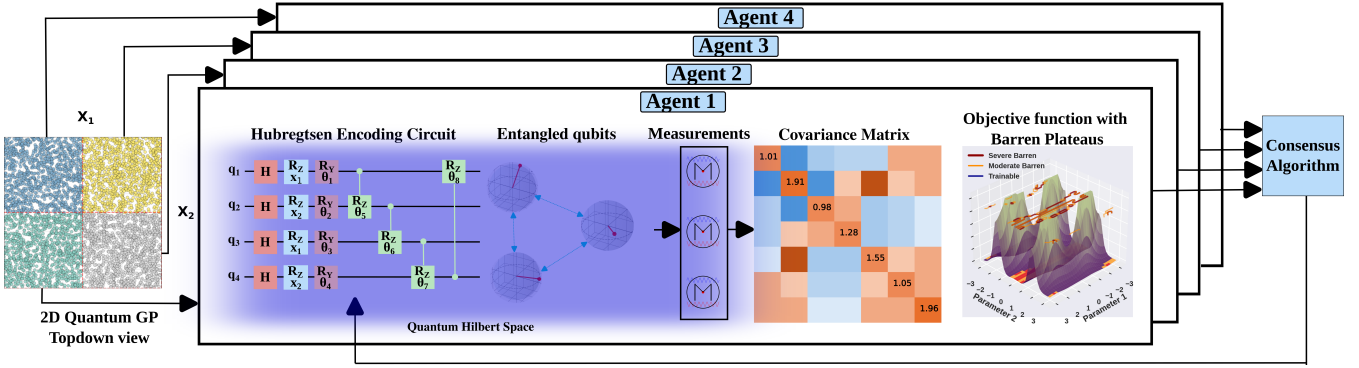}
	\centering
	\caption{The structure of the proposed DQGP with 4 agents. The consensus algorithm is the proposed DR-ADMM optimization.}
	\label{fig:approach}
\end{figure*}

\begin{problem}
Develop a \textbf{Distributed consensus Riemannian ADMM (DR-ADMM)} approach for optimizing quantum kernel hyperparameters across multiple agents. 
\end{problem}
\begin{problem}
Formulate a \textbf{Distributed Quantum Gaussian Process (DQGP)} algorithm that simultaneously overcomes the expressivity limitations of DGPs and the scalability challenges of QGPs.
\end{problem}

% \begin{problem} \label{problem3}
% Enhance the predictive capabilities of Gaussian processes training (solving Equation (\ref{eqn:gpopt}) under Assumption \ref{assumption:L}) with independent local agents (Assumption \ref{assumption:local}) using distributed quantum GP.
% \end{problem}

%%%%%%%%%%%%%%%%%%%%%%%%%%%%%%%%%%%%%%%%%%%%%%%%%%%%%%%%%%%%%%%%%%%%%%%%

\section{Proposed Methodology} \label{sec3}
In this section, we present the proposed methodologies Distributed consensus Riemannian ADMM (DR-ADMM) and Distributed Quantum Gaussian Process (DQGP).
\subsection{Distributed consensus Riemannian ADMM}
The ADMM-based consensus problem~\eqref{eq:apxGP} %Equations \eqref{eqn:admmz}, \eqref{eqn:admmt}, and \eqref{eqn:admmp} 
accounts for GP training with classical kernel hyperparameters that operate in a Euclidean parameter space. For quantum kernels, the hyperparameters are mostly rotational and thus lie in a more complex parameter space, a Riemannian manifold~\cite{lee2018introduction}. To address this, we first define the manifold in which the quantum hyperparameters exist and then describe the distributed consensus Riemannian ADMM algorithm. In~\cite{li2022riemannian}, a Riemannian ADMM algorithm for solving nonconvex problems in centralized topologies is introduced.

\subsubsection{Manifold Definition and Operations}
The torus manifold for quantum circuit parameters is expressed as~$\mathcal{T}^P = \mathcal{S}^1 \times \mathcal{S}^1 \times \dots \times \mathcal{S}^1$, % \text{($P$ times for $\#$ parameters)}$, 
where $\mathcal{S}^1$ is the circle manifold for rotational parameters $\theta_p \in [0, \pi]$ for all $p = 1, 2, \ldots, P$. A combination of rotational quantum gates with parameters in $[0, \pi]$ ensures full-coverage of Bloch sphere~\cite{glendinning2005bloch}. Our choice of $[0, \pi]$ narrows the search space without sacrificing the expressiveness associated with the standard $[0, 2\pi]$ range. Subsequent operations can be constructed over the torus manifold, % as,
\begin{align*}
    \text{\textit{Manifold Projection}}&: \, \Pi_{\mathcal{T}}(\theta) = \theta \bmod \pi,\\
    \text{\textit{Riemannian Distance}}&: \, d_{\mathcal{T}}(\theta_p, \theta_r) = \|\mathcal{W}(\theta_p - \theta_r)\|_2, \\
    %&\quad \mathcal{W}(\theta) = \left[\left(\theta + \frac{\pi}{2}\right) \bmod \pi\right] - \frac{\pi}{2},\\
    \textit{Retraction} &: \, \mathcal{R}_{\theta_{p}}(\theta_r) = \Pi_{\mathcal{T}}(\theta_p + \theta_r),\\
    \textit{Logarithmic Map} &: \, \text{Log}_{\theta_{p}}(\theta_r) = \Pi_{\mathcal{T}}(\theta_r - \theta_p),\\
    \textit{Vector Transport} &: \, \Gamma_{\theta_p \to \theta_r}(\mathbi{v}) = \mathbi{v},\\
    \textit{Inner Product}&: \, \inpt{\boldsymbol{\theta_p}}{\boldsymbol{\theta_r}} =  \inp{\mathcal{W}({\boldsymbol{\theta_p}})}{\mathcal{W}(\boldsymbol{\theta_r})},
\end{align*}
where $\mathcal{W}(\theta) = \left[\left(\theta + {\pi}/{2}\right) \bmod \pi\right] - {\pi}/{2}$. Vector transport operation transports the tangent vector $\mathbi{v}$~from point $\theta_p$ to $\theta_r$, which is an identity on a torus manifold. As the torus manifold is locally flat, Riemannian gradient is identical to the Euclidean gradient on a torus, i.e., $\text{grad}_{\mathcal{T}} f(\theta) = \nabla f(\theta)$. Moreover, we assume,
\begin{assumption}\label{assumption:bounded}
    All operations and gradients corresponding to the torus manifold are uniformly bounded.
\end{assumption}

\subsubsection{Quantum Gaussian Process Loss Function}
Let us recast the negative marginal log-likelihood function~\eqref{eqn:LL} for QGPs as,
%From Equation \eqref{eqn:LL}, we can 
\begin{equation*}
\mathcal{L}_{\mathrm{Q}}(\boldsymbol{\theta}) = \frac{1}{2} \Big( \mathbi{y}^\intercal \boldsymbol{\kappa}^{-1}_{\mathrm{Q}}(\mathbi{X},\mathbi{X}\vert \boldsymbol{\theta}) \mathbi{y} + \log \vert \boldsymbol{\kappa}_{\mathrm{Q}}(\mathbi{X},\mathbi{X}\vert \boldsymbol{\theta}) \vert + N \log(2\pi) \Big),
\end{equation*}
where $\boldsymbol{\kappa}_{\mathrm{Q}}(\mathbi{X},\mathbi{X}\vert \boldsymbol{\theta})$ is the quantum kernel matrix with intrinsic measurement noise. The derivatives of the quantum kernel with respect to the individual hyperparameter are computed using the \textit{parameter shift rule} \cite{wierichs2022general}, as the quantum gates are mostly rotational,
\begin{eqnarray*}
    \frac{\partial [\boldsymbol{\kappa}_{Q}]_{ij}}{\partial \theta_p} = \frac{ [\boldsymbol{\kappa}_{Q}]_{ij}(\boldsymbol{\theta} + \delta \mathbi{e}_p) -  [\boldsymbol{\kappa}_{Q}]_{ij}(\boldsymbol{\theta} - \delta \mathbi{e}_p)}{2\delta},
\end{eqnarray*}
where $\delta$ is the shift value and $\mathbi{e}_p$ is the $p$-th unit vector. Then, from~\eqref{eqn:gradient} the Quantum NLL loss gradient is described as,
\begin{eqnarray*}
\nabla_{\boldsymbol{\theta}}\mathcal{L}_{\mathrm{Q}}(\boldsymbol{\theta}) = \frac{1}{2} \sum_{p=1}^P \text{Tr}\left\{\left(\boldsymbol{\kappa}_{\mathrm{Q}}^{-1} - \boldsymbol{\kappa}_{\mathrm{Q}}^{-1}\mathbi{y}\mathbi{y}^T \boldsymbol{\kappa}_{\mathrm{Q}}^{-1}\right) \frac{\partial \boldsymbol{\kappa}_{\mathrm{Q}}}{\partial \theta_p}\right\}.
\end{eqnarray*}

% \textbf{Efficient Gradient Computation:}
% \begin{equation}
% \frac{\partial \mathcal{L}}{\partial \theta_k} = \frac{1}{2} \sum_{i,j} \left(C^{-1}_{ij} - (C^{-1}y)_i(C^{-1}y)_j\right) \frac{\partial C_{ij}}{\partial \theta_k}
% \end{equation}

\subsubsection{Distributed consensus Riemannian ADMM (DR-ADMM)}
The distributed optimization problem can be formulated as minimizing the sum of local QGP loss function across all agents, 
while enforcing a consensus constraint and assuming $L_p$-smooth cost functions, % that aligns the local model parameters with a shared global parameter,
$\min_{\{\boldsymbol{\theta}_m\}} \sum_{m=1}^M \mathcal{L}_{\mathrm{Q},m}(\boldsymbol{\theta}_m)$ subject to $\boldsymbol{\theta}_m = \mathbi{z}, \, \forall m$.
% \begin{align*}
%     \min_{\{\boldsymbol{\theta}_m\}} \sum_{m=1}^M \mathcal{L}_{\mathrm{Q},m}(\boldsymbol{\theta}_m) \quad \text{subject to } \boldsymbol{\theta}_m = \mathbi{z}, \, \forall m.
% \end{align*}
\begin{assumption}\label{assumption:L_psmooth}
    $\mathcal{L}_{Q,m}$ is ~$L_p$-smooth, i.e., $\mathcal{L}_{Q,m}(\boldsymbol{z}^{(s+1)}) - \mathcal{L}_{Q,m}(\boldsymbol{z}^{(s)}) \\ \le \inpt{\nabla \mathcal{L}_{Q,m}(\boldsymbol{z}^{(s)})}{\boldsymbol{z}^{(s+1)} - \boldsymbol{z}^{(s)}} + \frac{L_p}{2}\mathrm{d}_{\mathcal{T}}^2(\boldsymbol{z}^{(s+1)}, \boldsymbol{z}^{(s)}), \ \forall m \in [1,M]$.
\end{assumption}

To solve the distributed optimization problem, we construct an augmented Lagrangian function defined on a Riemannian manifold,
\begin{align*}
    \mathbb{L}_\rho(\boldsymbol{\theta}_m, \mathbi{z}, \boldsymbol{\psi}_m) &= \sum_{m=1}^M \bigg( \mathcal{L}_{\mathrm{Q},m}(\boldsymbol{\theta}_m) + \boldsymbol{\psi}_m^\intercal\text{Log}_\mathbi{z}(\boldsymbol{\theta}_m) + \frac{\rho}{2}\|\text{Log}_\mathbi{z}(\boldsymbol{\theta}_m)\|^2 \bigg).
\end{align*}

\begin{algorithm}[!t]
\caption{DR-ADMM}
\label{alg:DR-ADMM}
\begin{algorithmic}[1]
\Statex \textbf{Input:}
$\{\mathcal{D}_m=[\mathbi{X}_m, \mathbi{y}_m], \tilde{\boldsymbol{\theta}}^{(s)}_m=[(\boldsymbol{\theta}^{(s)}_m)^\intercal, \sigma_{\epsilon, m}^{2(s)}]^\intercal,$
\Statex \qquad \qquad $\tilde{\boldsymbol{\psi}}^{(s)}_m=[(\boldsymbol{\psi}^{(s)}_m)^\intercal, \psi^{(s)}_{\sigma^2_{\epsilon, m}}]^\intercal\}_{m=1}^{M}$, $\kappa_Q$, $\delta$, $\rho$, $\mathbi{L}$.% Let $\tilde{\boldsymbol{\theta}} = [\boldsymbol{\theta}^\top, \sigma^2]^\top \in \mathcal{M}$, where $\mathcal{M} = \mathcal{T}^P \times [\sigma^2_{\min}, 1]$ is the product manifold.
\Statex \textbf{Output:} 
$\tilde{\mathbi{z}}^{(s+1)}$, $\{\tilde{\boldsymbol{\theta}}^{(s+1)}_m\}^{M}_{m=1}$, $\{\tilde{\boldsymbol{\psi}}^{(s+1)}_m\}^{M}_{m=1}$, $\|\mathbi{r}_{\mathrm{pri}}^{(s)}\|_{2}$, $\|\mathbi{r}_{\mathrm{dual}}^{(s)}\|_{2}$
\State $\{\tilde{\boldsymbol{\varphi}}^{(s)}_m\}^M_{m=1} = \{\tilde{\boldsymbol{\theta}}_m^{(s)}\}^M_{m=1} + \frac{\{\tilde{\boldsymbol{\psi}}_m^{(s)}\}^M_{m=1}}{\rho}$
\State $\mathbi{z}_{\boldsymbol{\theta}}^{(s+1)} = \frac{1}{2} \Pi_{\mathcal{T}} \left(\text{atan2}\left[ \sum\limits_{m=1}^{M} \sin\left(2\boldsymbol{\varphi}_{\boldsymbol{\theta},m}^{(s)}\right), \sum\limits_{m=1}^{M} \cos\left(2\boldsymbol{\varphi}_{\boldsymbol{\theta},m}^{(s)}\right) \right]\right)$ \label{alg:circmean}
\State $z_{\sigma^2}^{(s+1)} = \frac{1}{M} \sum\limits_{m=1}^{M} \varphi_{\sigma^2_{\epsilon, m}}^{(s)}, \quad \tilde{\mathbi{z}}^{(s+1)} = [(\mathbi{z}_{\boldsymbol{\theta}}^{(s+1)})^\intercal, z_{\sigma^2_{\epsilon}}^{(s+1)}]^\intercal$
\For{$m = 1, \ldots, M$ \textbf{in PARALLEL}}
    \State Compute quantum kernel matrix: $[\boldsymbol{\kappa}_{Q}]_{m}(\mathbi{X}_m, \mathbi{X}_m | \mathbi{z}_{\boldsymbol{\theta}}^{(s+1)})$
    \For{$p = 1, 2, \dots, P$ \textbf{in PARALLEL}}
        \State Perform kernel evaluations $[\boldsymbol{\kappa}_{Q}]_m(\boldsymbol{\theta}^{(s)}_m \pm \delta \mathbi{e}_p | \mathbi{z}_{\boldsymbol{\theta}}^{(s+1)})$
    \EndFor
    \State Compute derivatives: $\frac{\partial [\boldsymbol{\kappa}_{Q}]_{m}}{\partial \boldsymbol{\theta}}$ and $\frac{\partial \mathcal{L}_{\mathrm{Q},m}}{\partial \sigma_\epsilon^2}$
    \State Compute local gradient: $\nabla_{\tilde{\boldsymbol{\theta}}}\mathcal{L}_{\mathrm{Q},m}(\tilde{\mathbi{z}}^{(s+1)})$
    \State Update: $\tilde{\boldsymbol{\theta}}_m^{(s+1)} = \mathcal{R}_{\tilde{\mathbi{z}}^{(s+1)}}\left(-\frac{\nabla_{\tilde{\boldsymbol{\theta}}}\mathcal{L}_{\mathrm{Q},m}(\tilde{\mathbi{z}}^{(s+1)}) + \tilde{\boldsymbol{\psi}}_m^{(s)}}{\rho + L_m}\right)$
    \State Update: $\tilde{\boldsymbol{\psi}}_m^{(s+1)} = \tilde{\boldsymbol{\psi}}_m^{(s)} + \rho \text{Log}_{\tilde{\mathbi{z}}^{(s+1)}}(\tilde{\boldsymbol{\theta}}_m^{(s+1)})$
\EndFor
\State Compute residuals: $\|\mathbi{r}_{\mathrm{pri}}^{(s)}\|_{2}$, $\|\mathbi{r}_{\mathrm{dual}}^{(s)}\|_{2}$
\end{algorithmic}\label{alg:dradmm}
\end{algorithm}

Subsequently, the optimization scheme of the DR-ADMM yields,
\begin{subequations}
\begin{alignat}{3}
    \mathbi{z}^{(s+1)} = \underset{w \in \mathcal{T}^P}{\arg\min} \sum_{m=1}^{M} d_{\mathcal{T}}^2 \left(w, \boldsymbol{\theta}_m^{(s)} + \frac{\boldsymbol{\psi}_m^{(s)}}{\rho} \right)\label{eqn:dradmmz}\\
    \boldsymbol{\theta}_m^{(s+1)} = \mathcal{R}_{\mathbi{z}^{(s+1)}}\left(-\frac{\nabla_{\boldsymbol{\theta}}\mathcal{L}_{\mathrm{Q},m}\left(\mathbi{z}^{(s+1)}\right) + \boldsymbol{\psi}_m^{(s)}}{\rho + L_m}\right)\label{eqn:dradmmt}\\
    \boldsymbol{\psi}_m^{(s+1)} = \boldsymbol{\psi}_m^{(s)} + \rho  \text{Log}_{\mathbi{z}^{(s+1)}}\left(\boldsymbol{\theta}_m^{(s+1)}\right),\label{eqn:dradmmp}
\end{alignat}
\label{eq:dradmm}%
\end{subequations}
where $\mathbi{z}^{(s+1)}$ is the global consensus parameter, $\boldsymbol{\psi}_m^{(s+1)}$ the dual variable, $\nabla_{\boldsymbol{\theta}}\mathcal{L}_{\mathrm{Q},m}(\mathbi{z}^{(s+1)})$ the Riemannian gradient of local loss at $\mathbi{z}^{(s+1)}$, $L_m$ the Lipschitz constant for agent $m$, and $\mathcal{R}_{\mathbi{z}^{(s+1)}}$ the retraction operator from global parameter $\mathbi{z}^{(s+1)}$. \eqref{eqn:dradmmz} uses the simplified Karcher mean~\cite{lim2012matrix} to compute $\mathbi{z}^{(s+1)}$ in non-Euclidean fashion. The optimization of DR-ADMM~\eqref{eq:dradmm} iterates until both the primal and the dual residuals fall below the predefined tolerance thresholds $\|\mathbi{r}_{\mathrm{pri}}^{(s)}\|_2 \leq \epsilon_{\mathrm{pri}}$ and $\|\mathbi{r}_{\mathrm{dual}}^{(s)}\|_2 \leq \epsilon_{\mathrm{dual}}$, respectively. The primal and dual residuals are computed as $\mathbi{r}_{\mathrm{pri}}^{(s)} = [ d_{\mathcal{T}}(\boldsymbol{\theta}_1^{(s)}, \mathbi{z}^{(s)}) \ \hdots \ d_{\mathcal{T}}(\boldsymbol{\theta}_M^{(s)}, \mathbi{z}^{(s)})]^{\intercal}$ and  $\mathbi{r}_{\mathrm{dual}}^{(s)} = \rho  d_{\mathcal{T}}(\mathbi{z}^{(s)}, \mathbi{z}^{(s-1)})$, respectively. The implementation details of DR-ADMM are presented in Algorithm~\ref{alg:dradmm}. In $\mathcal{T}^P$, the Karcher mean for $\mathbi{z}^{(s+1)}$ can be reduced to the circular mean (Alg.~\ref{alg:dradmm}-line~\ref{alg:circmean}). The noise hyperparameter $\sigma^2_\epsilon$ reaches consensus using the standard ADMM with Euclidean mean. The learned noise hyperparameter represents the observation noise and ensures that the quantum kernel matrix remains positive-definite and numerically stable for matrix inversion. %Next, we present the theorem stating the convergence of DR-ADMM with quantum hyperparameters $\boldsymbol{z}$.}

\begin{theorem}[%Global 
Convergence of DR-ADMM] 
Let the negative marginal log-likelihood functions $\mathcal{L}_{Q,m}: \mathcal{T}^p \to \mathbb{R}$ be $L_p$-smooth $\forall m \in [1,M]$ on the torus manifold $\mathcal{T}^p$ with bounded projections $\proj(\cdot)$, and assume the existence of a uniform bound $C < \infty$ such that $\normt{\nabla \mathcal{L}_{Q,m}(\cdot)} \leq C$. Then, for a sufficiently large penalty parameter $\rho > 0$, the sequence $\{\boldsymbol{\theta}_m^{(s)}, \boldsymbol{\psi}_m^{(s)}, \mathbi{z}^{(s)}\}$ generated by $s$ iterations of the Distributed consensus Riemannian ADMM algorithm converges %globally 
to a stationary point $(\boldsymbol{\theta}^*_m, \boldsymbol{\psi}^*_m, \boldsymbol{z}^*)$ that satisfies the KKT conditions for the consensus problem, characterized by: 
\begin{itemize} 
\item \textbf{Primal Feasibility:} The primal residuals vanish, \\$\lim_{s \to \infty} \dist(\boldsymbol{\theta}_m^{(s)}, \mathbi{z}^{(s)}) = 0$, yielding $\boldsymbol{\theta}_m^* = \mathbi{z}^* \, \forall m$. 
\item \textbf{Dual Feasibility and Stationarity:} The dual residuals vanish, $\lim_{s \to \infty} \rho \normt{\mathbi{z}^{(s)} - \mathbi{z}^{(s-1)}} = 0$, and the gradient of the negative marginal log-likelihood functions with respect to the consensus variable is zero, i.e., %resulting in the stationary condition
$\sum_{m=1}^M \nabla \mathcal{L}_{Q,m}(\mathbi{z}^*) = 0$.
\item \textbf{Convergence Rate:} The algorithm achieves a sublinear convergence rate of $\mathcal{O}(1/S)$, where $S$ is the total number of iterations. To achieve $\xi$-accuracy in consensus residual, the required number of iterations is $S = \mathcal{O}\left(\frac{(\rho + L_{max})(V^{(0)} - V^*)}{\xi}\right)$, where $L_{max} = \max_m L_m$, and $V^* = \inf_s V^{(s)}$ is the optimal value.
\end{itemize} 
\end{theorem}

\textit{Proof:} (Sketch) Under Assumptions \ref{assumption:bounded} and \ref{assumption:L_psmooth}% and \ref{assumption:rho}
, define a Lyapunov function $V$ for the $s^{th}$ iteration with dual variable regularization as,
\begin{align*}
    V^{(s)} = &\sum_{m=1}^{M} \left[ \mathcal{L}_{Q,m}(\mathbi{z}^{(s)}) + \inpt{\boldsymbol{\psi}_m^{(s)}}{\logmap{\mathbi{z}^{(s)}}{\boldsymbol{\theta}_m^{(s)}}} + \frac{\rho}{2} \dist^2(\mathbi{z}^{(s)}, \boldsymbol{\theta}_m^{(s)}) \right] \nonumber\\
    &+ \frac{1}{2\rho} \sum_{m=1}^{M} \normt{\boldsymbol{\psi}_m^{(s)}}^2.
    \end{align*}
Then prove that it is non-increasing, \begin{align}
    V^{(s+1)} - &V^{(s)} \leq - |\Lambda(\rho, L_m)|%\frac{1}{4(\rho + L_m)} \times 
    \nonumber\\
    &\times 
    \sum_{m=1}^M \left[ \normt{\nabla \mathcal{L}_{Q,m}(\mathbi{z}^{(s)}) + \boldsymbol{\psi}_m^{(s)}}^2 + \dist^2(\boldsymbol{\theta}_m^{(s)}, \mathbi{z}^{(s)})\right]. \label{eqn:lyapunovineq}
\end{align}
From \eqref{eqn:lyapunovineq}, show that the optimization variables $\boldsymbol{\theta}_m$, $\boldsymbol{\psi}_m$, $\mathbi{z}$, the primal residual ($\sum_{s=0}^{\infty} \sum_{m=1}^M \dist^2(\boldsymbol{\theta}_m^{(s)}, \mathbi{z}^{(s)})$) and the dual residual ($\sum_{s=1}^{\infty} \rho^2 \normt{\mathbi{z}^{(s)} - \mathbi{z}^{(s-1)}}^2$) series remain uniformly bounded during DR-ADMM iterations. Ultimately, in the limit, these residuals vanish, and the optimization variables converge to $\boldsymbol{\theta}^*_m, \boldsymbol{\psi}^*_m, \boldsymbol{z}^*$. Applying \eqref{eqn:admmt} at these limit points leads to $\sum_{m=1}^M \nabla \mathcal{L}_{Q,m}(\mathbi{z}^*) = 0$. The stated sublinear convergence rate can be proven using \eqref{eqn:lyapunovineq}.

\subsection{Distributed Quantum Gaussian Process}

\begin{algorithm}[!t]
\caption{DQGP}
\label{alg:distributed-qgp}
\begin{algorithmic}[1]
\Statex \textbf{Input:} 
$\mathcal{D} = \{\mathbi{X}, \mathbi{y}\}$, $M$, $\delta$, $\kappa_Q(\Phi(q, \iota))$, $\rho$, $\mathbi{L}$, $F$, $s_{max}$, $T$,  $\epsilon_{\mathrm{pri}}$, $\epsilon_{\mathrm{dual}}$
\Statex \textbf{Output:} 
$\tilde{\mathbi{z}}^*$, $\text{NLPD}_{\text{test}}$, $\text{NRMSE}_{\text{test}}$
\State \textbf{Initialize:} 
$\tilde{\boldsymbol{\theta}}_m^{(0)}, \tilde{\mathbi{z}}^{(0)}, \tilde{\boldsymbol{\psi}}_m^{(0)}$, $\kappa_Q(\tilde{\mathbi{z}}^{(0)})$, $NLPD^*_{CV}=\infty$, $t=0$, $s=0$
\State $\mathcal{D}_{\text{train}}, \mathcal{D}_{\text{test}} \leftarrow \texttt{TrainTestSplit}(\mathcal{D})$
\State $\{\mathcal{D}_{m}\}^{M}_{m=1}$ $\leftarrow$ \texttt{Regional}\texttt{k-dTreeSplit} $(\mathcal{D}_{\text{train}}, M)$ \label{alg:regional}
\While{$s < s_{max}$ and $t<T$}
    \State $\tilde{\mathbi{z}}^{(s+1)}, \{\tilde{\boldsymbol{\theta}}^{(s+1)}_m\}^M_{m=1}, \{\tilde{\boldsymbol{\psi}}^{(s+1)}_m\}^M_{m=1}, \|\mathbi{r}_{\mathrm{pri}}^{(s)}\|_{2}, \|\mathbi{r}_{\mathrm{dual}}^{(s)}\|_{2}$
    \Statex \qquad $\leftarrow$ $\texttt{DR-ADMM}(\{\mathcal{D}_m, \tilde{\boldsymbol{\theta}}_m^{(s)}, \tilde{\boldsymbol{\psi}}^{(s)}_m\}_{m=1}^{M}, \kappa_Q, \delta, \rho, \mathbi{L})$~\label{alg:DR-ADMM in DQGP}
    \State $NLPD_{CV} \leftarrow \texttt{F-fold\_Cross-Validation}(\tilde{\mathbi{z}}^{(s+1)}, \bigcup_m \mathcal{D}_m)$\label{alg:f-fold in DQGP}
    \State \textbf{if} {$\text{NLPD}_{\text{CV}} < \text{NLPD}^*_{\text{CV}}$} \textbf{then}
    \State \quad \ \ $\text{NLPD}^*_{\text{CV}} = \text{NLPD}_{\text{CV}}$, $\tilde{\mathbi{z}}^* = \tilde{\mathbi{z}}^{(s+1)}$, $t = 0$
    \State \textbf{else} $t = t + 1$
    \State \textbf{end if}
    \State \textbf{if} {$\|\mathbi{r}_{\mathrm{pri}}^{(s)}\|_2 \leq \epsilon_{\mathrm{pri}}$ \textbf{and} $\|\mathbi{r}_{\mathrm{dual}}^{(s)}\|_2 \leq \epsilon_{\mathrm{dual}}$} \textbf{then break}
    \State $s=s+1$
\EndWhile
\State NLPD$_{\text{test}}$, NRMSE$_{\text{test}}$ $\leftarrow \texttt{gPoE}(\{\mathcal{D}_m\}_{m=1}^M, \mathcal{D}_{\text{test}}, \tilde{\mathbi{z}}^*)$\label{alg:test pred in DQGP} \Comment{\cite{deisenroth2015distributed}}
\end{algorithmic}
\end{algorithm}

%To effectively leverage the quantum Hilbert space for the classical datasets, it is imperative to develop methodologies that flexibly interface between these domains. Quantum encoding circuits, as introduced earlier, provide a way to map the classical data into the quantum states. Subsequently, the quantum kernel functions are designed to model relations between the encoded quantum states and, equivalently, the original decoded classical data. 

\subsubsection{Quantum Encoding Circuits}
%The formulation of the quantum encoding circuits is an active field of research. 
%In this work, w
We primarily employ %the following encoding circuits based on our research objectives.
% two types of quantum encoding circuits. %tailored to our research objectives.
% \begin{figure}[h]
%   \centering
%   \includegraphics[width=\linewidth]{encoding_circuit_visualizer/yz_cx_encoding.png}
%   \caption{YZ-CX Encoding Circuit with layers $\iota = 2$ and qubits $q = 4$ for $2D \, \mathbi{x} = \{x_1, x_2\}$.}
%   \label{fig:yz_cx}
%   \Description{YZ_CX.}
% \end{figure} 
% \begin{figure}[h]
%   \centering
%   \includegraphics[width=.8\columnwidth]{figures/chebyshev.png}
%   \caption{Chebyshev Parameterized Quantum Circuit with layers $\iota = 1$ and qubits $q = 4$ for $2D \, \mathbi{x} = \{x_1, x_2\}$ resulting in $16$ trainable parameters.}
%   \label{fig:chebyshevPQC}
%   \Description{Chebyshev PQC.}
% \end{figure}
% The first is the Chebyshev parameterized quantum circuit \cite{kreplin2024reduction, williams2023quantum} %is composed 
% which consists of $R_Y, R_X$, and conditional-$R_Z$ quantum rotational gates. %~\eqref{eqn:quantumgates}. % as shown in Figure \ref{fig:chebyshevPQC}.
% The second is 
the Hubregtsen encoding circuit~\cite{hubregtsen2022training}, %alternatively, 
which contains the Hadamard gate $H$, rotation gates $R_Z$, $R_Y$, and conditional-$R_Z$ as shown in Fig.~\ref{fig:approach}. The conditional quantum gates play a crucial role by inducing entanglement among the qubits, whereas the Hadamard gate %gives rise to 
introduces superposition. %effects 
These are essential properties for representing complex quantum systems. %, another critical feature of quantum systems. 
The overall configuration of these circuits depends on %Depending on 
the number of qubits $q$, the number of layers $\iota$, and the dimensionality $D$ of the %classical input data
input space. Accordingly, the encoding circuits have $P$ parameters %denoted as 
$\{\theta_1, \theta_2, \ldots, \theta_P \} \in %\mathbb{R}^P
[0, \pi]^P$. %Next, we focus on the construction of quantum kernel functions. %, which form the foundation of the proposed DQGP framework.

\begin{table*}[!t]
\caption{%Comparing the uncertainty estimation 
NLPD$_{\mathrm{test}}$ and %predictive accuracy 
NRMSE$_{\mathrm{test}}$ %of our approach with other methods using the same dataset configurations and 
for $N=500$. %for the 2D Quantum Gaussian process prior, hyperparameter 
Set 1 $= \{0.58, 2.45, 1.88, 1.40, 0.31, 1.44\}$ and %hyperparameter 
Set 2 $=\{1.18, 2.99, 2.30, 1.88, 0.49, 0.49\}$.}
\centering
\small{\begin{tabular}{c c c c c c c c c c c}
\toprule
  & & \multicolumn{2}{c}{Single Agent Method} & & \multicolumn{6}{c}{Distributed Methods} \\
  \cmidrule(lr){3-4} \cmidrule(lr){6-11}
 {Data} &  Subset & \multicolumn{2}{c}{Full-GP \cite{williams2006gaussian}} & $M$ & \multicolumn{2}{c}{DQGP-DR-ADMM} & \multicolumn{2}{c}{FACT-GP \cite{deisenroth2015distributed}} & \multicolumn{2}{c}{apxGP-ADMM \cite{xie2019distributed}} \\
 \cmidrule(lr){3-4} \cmidrule(lr){6-7} \cmidrule(lr){8-9} \cmidrule(lr){10-11}
  & & $NLPD_{\text{test}}\downarrow$ & $NRMSE_{\text{test}}\downarrow$ &  & $NLPD_{\text{test}}\downarrow$ & $NRMSE_{\text{test}}\downarrow$ & $NLPD_{\text{test}}\downarrow$ & $NRMSE_{\text{test}}\downarrow$ & $NLPD_{\text{test}}\downarrow$ & $NRMSE_{\text{test}}\downarrow$ \\
\midrule
 \multirow{12}{2em}{S\\R\\T\\M} & \multirow{3}{*}{N17E073} & \multirow{3}{*}{$-0.96 \pm 0.080$} & \multirow{3}{*}{$0.10 \pm 0.011$} & {4} & $\boldsymbol{-0.30 \pm 0.320}$ & $\boldsymbol{0.17 \pm 0.021}$ & $-0.02 \pm 0.194$ & $0.18 \pm 0.012$ & $-0.04 \pm 0.187$ & $0.18 \pm 0.012$ \\
 & & & & {8} & $\boldsymbol{-0.40 \pm 0.203}$ & $\boldsymbol{0.17 \pm 0.026}$ & $0.33 \pm 0.243$ & $0.21 \pm 0.012$ & $0.37 \pm 0.260$ & $0.21 \pm 0.012$ \\
 & & & & {27} & $\boldsymbol{-0.09 \pm 0.136}$ & $0.25 \pm 0.032$ & $0.49 \pm 0.283$ & $\boldsymbol{0.22 \pm 0.013}$ & $0.54 \pm 0.272$ & $\boldsymbol{0.22 \pm 0.013}$ \\
 \cmidrule{2-11}
 & \multirow{3}{*}{N43W080} & \multirow{3}{*}{$-1.46 \pm 0.122$} & \multirow{3}{*}{$0.06 \pm 0.016$} & {4} & $-0.46 \pm 0.684$ & $\boldsymbol{0.13 \pm 0.020}$ & $\boldsymbol{-0.49 \pm 0.205}$ & $0.16 \pm 0.020$ & $-0.46 \pm 0.202$ & $0.16 \pm 0.018$ \\
 & & & & {8} & $\boldsymbol{-0.57 \pm 0.421}$ & $\boldsymbol{0.13 \pm 0.016}$ & $-0.37 \pm 0.224$ & $0.17 \pm 0.023$ & $-0.30 \pm 0.218$ & $0.17 \pm 0.021$ \\
 & & & & {27} & $\boldsymbol{-0.28 \pm 0.240}$ & $0.21 \pm 0.035$ & $-0.17 \pm 0.244$ & $\boldsymbol{0.18 \pm 0.024}$ & $-0.15 \pm 0.249$ & $\boldsymbol{0.18 \pm 0.023}$ \\
 \cmidrule{2-11}
 & \multirow{3}{*}{N45W123} & \multirow{3}{*}{$-1.00 \pm 0.245$} & \multirow{3}{*}{$0.10 \pm 0.020$} & {4} & $\boldsymbol{-0.41 \pm 0.399}$ & $\boldsymbol{0.15 \pm 0.020}$ & $0.05 \pm 0.430$ & $0.19 \pm 0.022$ & $0.01 \pm 0.437$ & $0.18 \pm 0.022$ \\
 & & & & {8} & $\boldsymbol{-0.45 \pm 0.445}$ & $\boldsymbol{0.15 \pm 0.023}$ & $0.33 \pm 0.468$ & $0.20 \pm 0.021$ & $0.37 \pm 0.482$ & $0.21 \pm 0.022$ \\
 & & & & {27} & $\boldsymbol{-0.24 \pm 0.171}$ & $0.22 \pm 0.030$ & $0.44 \pm 0.487$ & $\boldsymbol{0.21 \pm 0.020}$ & $0.52 \pm 0.500$ & $0.22 \pm 0.020$ \\
 \cmidrule{2-11}
 & \multirow{3}{*}{N47W124} & \multirow{3}{*}{$-0.68 \pm 0.133$} & \multirow{3}{*}{$0.13 \pm 0.018$} & {4} & $\boldsymbol{-0.55 \pm 0.188}$ & $\boldsymbol{0.15 \pm 0.019}$ & $1.14 \pm 0.609$ & $0.25 \pm 0.022$ & $1.43 \pm 0.595$ & $0.26 \pm 0.021$ \\
 & & & & {8} & $\boldsymbol{-0.57 \pm 0.173}$ & $\boldsymbol{0.16 \pm 0.023}$ & $1.54 \pm 0.648$ & $0.27 \pm 0.021$ & $1.64 \pm 0.630$ & $0.27 \pm 0.020$ \\
 & & & & {27} & $\boldsymbol{-0.04 \pm 0.177}$ & $\boldsymbol{0.26 \pm 0.039}$ & $1.79 \pm 0.672$ & $0.28 \pm 0.020$ & $1.83 \pm 0.675$ & $0.28 \pm 0.020$ \\
\midrule 
\multirow{6}{2em}{2D QGP prior} & \multirow{3}{*}{Set 1} & \multirow{3}{*}{$-0.35 \pm 0.200$} & \multirow{3}{*}{$0.04 \pm 0.006$} & {4} & $\boldsymbol{0.10 \pm 0.155}$ & $\boldsymbol{0.05 \pm 0.010}$ & $2.98 \pm 4.286$ & $0.09 \pm 0.044$ & $3.54 \pm 5.279$ & $0.09 \pm 0.044$ \\
 & & & & {8} & $\boldsymbol{0.24 \pm 0.130}$ & $\boldsymbol{0.07 \pm 0.012}$ & $7.87 \pm 8.547$ & $0.13 \pm 0.057$ & $8.31 \pm 8.816$ & $0.13 \pm 0.057$ \\
 & & & & {27} & $\boldsymbol{0.75 \pm 0.196}$ & $\boldsymbol{0.13 \pm 0.026}$ & $15.61 \pm 10.412$ & $0.18 \pm 0.051$ & $15.78 \pm 10.690$ & $0.18 \pm 0.050$ \\
 \cmidrule{2-11}
 & \multirow{3}{*}{Set 2} & \multirow{3}{*}{$-0.37 \pm 0.209$} & \multirow{3}{*}{$0.04 \pm 0.009$} & {4} & $\boldsymbol{0.06 \pm 0.147}$ & $\boldsymbol{0.06 \pm 0.011}$ & $1.37 \pm 2.871$ & $0.08 \pm 0.044$ & $2.18 \pm 5.522$ & $0.08 \pm 0.047$ \\
 & & & & {8} & $\boldsymbol{0.27 \pm 0.163}$ & $\boldsymbol{0.08 \pm 0.013}$ & $21.08 \pm 47.463$ & $0.15 \pm 0.108$ & $20.26 \pm 47.008$ & $0.14 \pm 0.106$ \\
 & & & & {27} & $\boldsymbol{0.64 \pm 0.159}$ & $\boldsymbol{0.13 \pm 0.015}$ & $26.80 \pm 48.198$ & $0.19 \pm 0.099$ & $26.15 \pm 46.964$ & $0.19 \pm 0.098$ \\
\bottomrule
\end{tabular}
}
\label{tab:results_500}
\end{table*}

\subsubsection{Quantum Kernels}
A standard choice for a quantum kernel based on the quantum fidelity measure $\mathcal{F}$~\eqref{eqn:fidelityker} %that is widely used in quantum computing, can be expressed as,
is, $\kappa_{\mathcal{F}}(\mathbi{x}, \mathbi{x}'| \boldsymbol{\theta}) = \left| \langle \Phi(\mathbi{x}, \boldsymbol{\theta})|\Phi(\mathbi{x}',\boldsymbol{\theta}) \rangle \right|^2 = \left|\langle 0^{\otimes q}| U^\dagger(\mathbi{x},\boldsymbol{\theta}) U(\mathbi{x}', \boldsymbol{\theta}) |0^{\otimes q}\rangle\right|^2,$
% \begin{align*}%\label{eq:fidelity_kernel}
%     \kappa_{\mathcal{F}}(\mathbi{x}, \mathbi{x}'| \boldsymbol{\theta}) = \left| \langle \Phi(\mathbi{x}, \boldsymbol{\theta})|\Phi(\mathbi{x}',\boldsymbol{\theta}) \rangle \right|^2 \nonumber =\left|\langle 0^{\otimes q}| U^\dagger(\mathbi{x},\boldsymbol{\theta}) U(\mathbi{x}', \boldsymbol{\theta}) |0^{\otimes q}\rangle\right|^2,
% \end{align*}
where $U$ denotes the quantum encoding circuit with $q$ qubits and $\boldsymbol{\theta}$ %as its 
the hyperparameters. %However, because 
Since $\kappa_\mathcal{F}$ %~\eqref{eq:fidelity_kernel} 
lacks %due to the absence of 
observable-dependent operations, %$\kappa_\mathcal{F}$~\eqref{eq:fidelity_kernel} has limited 
its expressivity and modularity are limited. To manage this limitation, we %mainly utilize 
employ the \textit{Projected Quantum Kernel} (PQK) \cite{gil2024expressivity}. %Through 
By incorporating measurements \cite{heinosaari2008notes} based on observables, PQK maps the quantum states into a classical feature vector space and then applies the classical outer kernel on that projected space, % to compute the correlations,
% $\kappa_{\textrm{PQK}}(\mathbi{x}, \mathbi{x}'| \boldsymbol{\theta}) = \kappa_{\text{outer}}(\langle O \rangle_{\psi(\mathbi{x})}, \langle O \rangle_{\psi(\mathbi{x}')}),$
\begin{equation}\label{eqn:pqk}
\kappa_{PQK}(\mathbi{x}, \mathbi{x}'| \boldsymbol{\theta}) = \kappa_{\text{outer}}(\langle O \rangle_{\psi(\mathbi{x})}, \langle O \rangle_{\psi(\mathbi{x}')})
\end{equation}
where $\langle O \rangle_{\psi(\mathbi{x})} = \langle 0^{\otimes q} | U^\dagger(\mathbi{x}, \boldsymbol{\theta}) O U(\mathbi{x},\boldsymbol{\theta}) | 0^{\otimes q} \rangle$.
% \begin{equation*}
% \langle O \rangle_{\psi(\mathbi{x})} = \langle 0^{\otimes q} | U^\dagger(\mathbi{x}, \boldsymbol{\theta}) O U(\mathbi{x},\boldsymbol{\theta}) | 0^{\otimes q} \rangle.
% \end{equation*}
The most common choices for the observable operator $O$ include: (i) Pauli-Z measurement $O_Z = \bigotimes_{i=1}^q (\boldsymbol{\sigma}_{\mathbf{Z}})_i $, which captures global correlation among $q$ qubits; (ii) Local Pauli measurement $O_{\text{local}} = \{(\boldsymbol{\sigma}_{\mathbf{Z}})_1, %(\boldsymbol{\sigma}_{\mathbf{Z}})_2, 
\ldots, (\boldsymbol{\sigma}_{\mathbf{Z}})_{q}\}$, which measure individual qubits; and (iii) Mixed Pauli measurement $O_{\text{mixed}} = \{(\boldsymbol{\sigma}_{\mathbf{Z}})_1 (\boldsymbol{\sigma}_{\mathbf{Z}})_2, (\boldsymbol{\sigma}_{\mathbf{X}})_1 (\boldsymbol{\sigma}_{\mathbf{Y}})_2, \ldots\}$, which measure pairwise qubit correlations. %For enhanced expressivity, the \textit{multi-observable} PQK~\cite{huang2021power} can be employed where several observables are measured simultaneously. 
% This form of PQK with $o$ observables yields,
% \begin{align}\label{eqn:pqk}
% \kappa_{\textrm{mPQK}}(\mathbi{x}, \mathbi{x}'| \boldsymbol{\theta}) = \kappa_{\text{outer}}(\langle \mathbi{O} \rangle_{\psi(\mathbi{x})}, \langle \mathbi{O} \rangle_{\psi(\mathbi{x}')}), 
% \end{align}
%where $\langle \mathbi{O} \rangle_{\psi(\mathbi{x})} = \left[\langle O_1 \rangle_{\psi(\mathbi{x})}, \langle O_2 \rangle_{\psi(\mathbi{x})}, \ldots, \langle O_o \rangle_{\psi(\mathbi{x})}\right]^\intercal$. 
Any classical kernel can be used for the outer kernel $\kappa_{\text{outer}}$~\cite{manzhos2024degeneration}. %Common options~\cite{manzhos2024degeneration} include the Gaussian ($\gamma$-parameter) kernel %,
% \begin{equation}\label{eqn:gaussian}
% %\text{Gaussian} &: 
% \kappa^{G}_{\text{outer}}(\langle \mathbi{O} \rangle_{\psi(\mathbi{x})}, \langle \mathbi{O} \rangle_{\psi(\mathbi{x}')})=   \exp\left(-\gamma |\langle \mathbi{O} \rangle_{\psi(\mathbi{x})} - \langle \mathbi{O} \rangle_{\psi(\mathbi{x}')}|^2\right),% \, \text{with parameter $\gamma$}
% \end{equation}
%and the Mat\'ern kernel. %,
% \begin{align}\label{eqn:matern}
% %\text{Matérn \cite{manzhos2024degeneration}} &: 
% \kappa^{M}_{\text{outer}}(\langle \mathbi{O} \rangle_{\psi(\mathbi{x})}, \langle \mathbi{O} \rangle_{\psi(\mathbi{x}')}) \nonumber &= \frac{2^{1-\nu}}{\Gamma(\nu)} \left(\sqrt{2\nu} \frac{|\langle \mathbi{O} \rangle_{\psi(\mathbi{x})} - \langle \mathbi{O} \rangle_{\psi(\mathbi{x}')}|}{\ell}\right)^\nu \nonumber \\ \times~ K_\nu&\left(\sqrt{2\nu} \frac{|\langle \mathbi{O} \rangle_{\psi(\mathbi{x})} - \langle \mathbi{O} \rangle_{\psi(\mathbi{x}')}|}{\ell}\right),
% \end{align}
%where $\Gamma$ is the gamma function, $K_{\nu}$ the $\nu$-order modified Bessel function of second kind, $\nu > 0$ the smoothness, and $l > 0$ the lengthscale hyperparameter. 
With PQK, the computational cost is significantly reduced as it only requires %we only need to evaluate 
evaluating $\mathcal{O}(\text{card}(O))$ expectations, %whereas 
while the fidelity kernel involves computation of $\mathcal{O}(2^q)$ state overlap. In addition, the observable operator measurements provide a physical interpretability with respect to the feature vector space and a way to implement them on quantum hardware. The implementation details of the proposed \textit{Distributed Quantum Gaussian Process} are presented in Algorithm~\ref{alg:distributed-qgp}.

\begin{table*}[!t]
\caption{%Comparing the uncertainty estimation 
NLPD$_{\mathrm{test}}$ and %predictive accuracy 
NRMSE$_{\mathrm{test}}$ %of our approach with other methods using the same dataset configurations and 
for $N=5,000$. %for the 2D Quantum Gaussian process prior, hyperparameter 
Set 1 $= \{0.58, 2.45, 1.88, 1.40, 0.31, 1.44\}$ and %hyperparameter 
 Set 2 $=\{1.18, 2.99, 2.30, 1.88, 0.49, 0.49\}$.}
\centering
\small{\begin{tabular}{c c c c c c c c c c c}
\toprule
  & & \multicolumn{2}{c}{Single Agent Method} & & \multicolumn{6}{c}{Distributed Methods} \\
  \cmidrule(lr){3-4} \cmidrule(lr){6-11}
 {Data} &  Subset & \multicolumn{2}{c}{Full-GP \cite{williams2006gaussian}} & $M$ & \multicolumn{2}{c}{DQGP-DR-ADMM} & \multicolumn{2}{c}{FACT-GP \cite{deisenroth2015distributed}} & \multicolumn{2}{c}{apxGP-ADMM \cite{xie2019distributed}} \\
 \cmidrule(lr){3-4} \cmidrule(lr){6-7} \cmidrule(lr){8-9} \cmidrule(lr){10-11}
  & & $NLPD_{\text{test}}\downarrow$ & $NRMSE_{\text{test}}\downarrow$ &  & $NLPD_{\text{test}}\downarrow$ & $NRMSE_{\text{test}}\downarrow$ & $NLPD_{\text{test}}\downarrow$ & $NRMSE_{\text{test}}\downarrow$ & $NLPD_{\text{test}}\downarrow$ & $NRMSE_{\text{test}}\downarrow$ \\
\midrule
 \multirow{12}{2em}{S\\R\\T\\M} & \multirow{3}{*}{N17E073} & \multirow{3}{*}{$-1.27 \pm 0.030$} & \multirow{3}{*}{$0.07 \pm 0.004$} & {4} & $\boldsymbol{-0.96 \pm 0.199}$ & $\boldsymbol{0.10 \pm 0.014}$ & $0.23 \pm 0.117$ & $0.19 \pm 0.006$ & $0.24 \pm 0.116$ & $0.19 \pm 0.006$ \\
 & & & & {8} & $\boldsymbol{-0.94 \pm 0.194}$ & $\boldsymbol{0.10 \pm 0.015}$ & $0.60 \pm 0.164$ & $0.21 \pm 0.008$ & $0.63 \pm 0.141$ & $0.21 \pm 0.007$ \\
 & & & & {27} & $\boldsymbol{-0.85 \pm 0.031}$ & $\boldsymbol{0.11 \pm 0.005}$ & $0.79 \pm 0.162$ & $0.22 \pm 0.008$ & $0.83 \pm 0.170$ & $0.22 \pm 0.008$ \\
 \cmidrule{2-11}
 & \multirow{3}{*}{N43W080} & \multirow{3}{*}{$-1.78 \pm 0.017$} & \multirow{3}{*}{$0.03 \pm 0.001$} & {4} & $\boldsymbol{-1.19 \pm 0.079}$ & $\boldsymbol{0.07 \pm 0.003}$ & $-0.47 \pm 0.083$ & $0.14 \pm 0.006$ & $-0.48 \pm 0.076$ & $0.14 \pm 0.006$ \\
 & & & & {8} & $\boldsymbol{-1.22 \pm 0.110}$ & $\boldsymbol{0.07 \pm 0.006}$ & $-0.30 \pm 0.092$ & $0.15 \pm 0.007$ & $-0.28 \pm 0.091$ & $0.15 \pm 0.007$ \\
 & & & & {27} & $\boldsymbol{-1.00 \pm 0.093}$ & $\boldsymbol{0.09 \pm 0.005}$ & $-0.23 \pm 0.101$ & $0.16 \pm 0.007$ & $-0.15 \pm 0.110$ & $0.16 \pm 0.008$ \\
 \cmidrule{2-11}
 & \multirow{3}{*}{N45W123} & \multirow{3}{*}{$-1.45 \pm 0.031$} & \multirow{3}{*}{$0.06 \pm 0.003$} & {4} & $\boldsymbol{-1.03 \pm 0.109}$ & $\boldsymbol{0.09 \pm 0.006}$ & $-0.13 \pm 0.108$ & $0.17 \pm 0.007$ & $-0.12 \pm 0.108$ & $0.17 \pm 0.007$ \\
 & & & & {8} & $\boldsymbol{-0.99 \pm 0.217}$ & $\boldsymbol{0.09 \pm 0.011}$ & $0.08 \pm 0.113$ & $0.18 \pm 0.008$ & $0.18 \pm 0.124$ & $0.19 \pm 0.008$ \\
 & & & & {27} & $\boldsymbol{-0.93 \pm 0.074}$ & $\boldsymbol{0.10 \pm 0.007}$ & $0.26 \pm 0.127$ & $0.19 \pm 0.008$ & $0.35 \pm 0.133$ & $0.20 \pm 0.008$ \\
 \cmidrule{2-11}
 & \multirow{3}{*}{N47W124} & \multirow{3}{*}{$-1.22 \pm 0.043$} & \multirow{3}{*}{$0.08 \pm 0.006$} & {4} & $\boldsymbol{-0.99 \pm 0.080}$ & $\boldsymbol{0.10 \pm 0.007}$ & $0.84 \pm 0.112$ & $0.23 \pm 0.013$ & $0.84 \pm 0.115$ & $0.23 \pm 0.013$ \\
 & & & & {8} & $\boldsymbol{-0.95 \pm 0.091}$ & $\boldsymbol{0.10 \pm 0.010}$ & $1.05 \pm 0.119$ & $0.24 \pm 0.013$ & $1.05 \pm 0.121$ & $0.24 \pm 0.013$ \\
 & & & & {27} & $\boldsymbol{-0.86 \pm 0.055}$ & $\boldsymbol{0.11 \pm 0.006}$ & $1.18 \pm 0.127$ & $0.25 \pm 0.014$ & $1.17 \pm 0.129$ & $0.25 \pm 0.014$ \\
 \midrule
 \multirow{6}{2em}{2D QGP prior} & \multirow{3}{*}{Set 1} & \multirow{3}{*}{$-0.32 \pm 0.197$} & \multirow{3}{*}{$0.03 \pm 0.007$} & {4} & $\boldsymbol{-0.24 \pm 0.157}$ & $\boldsymbol{0.03 \pm 0.005}$ & $0.35 \pm 1.293$ & $0.05 \pm 0.025$ & $1.05 \pm 3.948$ & $0.05 \pm 0.031$ \\
 & & & & {8} & $\boldsymbol{-0.05 \pm 0.138}$ & $\boldsymbol{0.03 \pm 0.003}$ & $8.70 \pm 22.303$ & $0.10 \pm 0.060$ & $9.71 \pm 27.306$ & $0.10 \pm 0.068$ \\
 & & & & {27} & $\boldsymbol{0.14 \pm 0.161}$ & $\boldsymbol{0.05 \pm 0.004}$ & $14.76 \pm 25.471$ & $0.13 \pm 0.058$ & $14.83 \pm 27.800$ & $0.13 \pm 0.064$ \\
 \cmidrule{2-11}
 & \multirow{3}{*}{Set 2} & \multirow{3}{*}{$-0.38 \pm 0.272$} & \multirow{3}{*}{$0.03 \pm 0.010$} & {4} & $\boldsymbol{-0.22 \pm 0.198}$ & $\boldsymbol{0.03 \pm 0.005}$ & $1.35 \pm 3.251$ & $0.06 \pm 0.032$ & $1.57 \pm 3.723$ & $0.06 \pm 0.032$ \\
 & & & & {8} & $\boldsymbol{-0.14 \pm 0.151}$ & $\boldsymbol{0.03 \pm 0.004}$ & $6.82 \pm 12.556$ & $0.09 \pm 0.050$ & $5.58 \pm 8.786$ & $0.09 \pm 0.043$ \\
 & & & & {27} & $\boldsymbol{0.11 \pm 0.125}$ & $\boldsymbol{0.05 \pm 0.007}$ & $11.17 \pm 17.335$ & $0.11 \pm 0.057$ & $11.55 \pm 15.325$ & $0.12 \pm 0.053$ \\
\bottomrule
\end{tabular}
}
\label{tab:results_5000}
\end{table*}

In Fig.~\ref{fig:application}, we depict the conceptual architecture of DQGP, where each agent is a computational node assigned to a local subset of the spatial dataset and trains a local QGP model. After each local step, agents push a learned~$\boldsymbol{\tilde{\theta}}_m$ to a central server, which aggregates them into a consensus model $\tilde{\mathbi{z}}$. Next, each agent pulls $\tilde{\boldsymbol{z}}$ from the central server to form a consensus on their local models with the globally aggregated model. This push-pull mechanism produces model-level consensus. % and collaboration.}
In Algorithm~\ref{alg:distributed-qgp}, we use the negative log predictive density (NLPD) for validation, 
\begin{align}\label{eqn:nlpdtest}
\text{NLPD}_{\text{test}} = & \frac{1}{\text{card}(\mathcal{D}_\text{test})} \sum_{j \in \mathcal{D}_\text{test}} \left[\frac{1}{2}\log(2\pi\sigma^2_*(j)) + \frac{(y_j - \mu_*(j))^2}{2\sigma^2_*(j)}\right].
\end{align}
To evaluate the global parameter $\tilde{\mathbi{z}}$ across training iterations, we perform \texttt{F-fold\_Cross-Validation} using combined datasets $\bigcup_m \mathcal{D}_m$. For each fold $f = 1,2,\ldots,F$, the combined dataset is randomly shuffled and partitioned into training $\mathcal{D}^{f}_{\text{train}}$ and validation $\mathcal{D}^{f}_{\text{val}}$ sets to compute $\text{NLPD}^{f}_{\text{CV}}$. The mean of all these folds, $\text{NLPD}_{\text{CV}}$, is monitored over iterations to obtain the optimal global $\tilde{\mathbi{z}}^{*}$. For DQGP prediction at an unknown input $\boldsymbol{x}_*$, each agent $m$ computes its predictive mean $\mu_m(\boldsymbol{x}_*)$ and variance $\sigma_m^2(\boldsymbol{x}_*)$ using local dataset $\mathcal{D}_m$ and optimized consensus parameter $\tilde{\boldsymbol{z}}^*=[(\mathbi{z}_{\boldsymbol{\theta}}^{*})^\intercal, z_{\sigma_\epsilon^2}^{*}]^\intercal$ by,
\begin{align*}
    \mu_m(\mathbi{x}_*) &= [\boldsymbol{\kappa}_{Q}]_m(\mathbi{x}_*, \mathbi{X}_m|\mathbi{z}_{\boldsymbol{\theta}}^{*})
    \left([\boldsymbol{\kappa}_{Q}]_m(\mathbi{X}_m, \mathbi{X}_m|\mathbi{z}_{\boldsymbol{\theta}}^{*}) + z_{\sigma_\epsilon^2}^{*}\, \mathbi{I}_{N_m}\right)^{-1} \mathbi{y}_m,\\
    \sigma_m^2(\mathbi{x}_*) &= \kappa_{Q,m}(\mathbi{x}_*, \mathbi{x}_*|\mathbi{z}_{\boldsymbol{\theta}}^{*}) - [\boldsymbol{\kappa}_{Q}]_m(\mathbi{x}_*, \mathbi{X}_m|\mathbi{z}_{\boldsymbol{\theta}}^{*})\\
    &\times\left([\boldsymbol{\kappa}_{Q}]_m(\mathbi{X}_m, \mathbi{X}_m|\mathbi{z}_{\boldsymbol{\theta}}^{*}) + z_{\sigma_\epsilon^2}^{*}\, \mathbi{I}_{N_m}\right)^{-1}[\boldsymbol{\kappa}_{Q}]_m(\mathbi{X}_m, \mathbi{x}_*|\mathbi{z}_{\boldsymbol{\theta}}^{*}).
\end{align*}
The global aggregated prediction $\mu_*(\boldsymbol{x}_*)$ and variance $\sigma_*(\boldsymbol{x}_*)$ are then calculated using the Generalized Product of Experts (gPoE)~\cite{deisenroth2015distributed}. % with uniform weights $\beta_m = 1/M$.} 
Finally, we employ  the normalized root mean squared error (NRMSE) to assess the prediction, %performance,
\begin{eqnarray}
\text{NRMSE}_{\text{test}} = \frac{\sqrt{\frac{1}{\text{card}(\mathcal{D}_\text{test})} \sum_{j \in \mathcal{D}_\text{test}} (y_j - \mu_*(j))^2}}{|y_{max} - y_{min}|}. \label{eqn:nrmsetest}
\end{eqnarray}

\section{Numerical Experiments \& Results}\label{sec4}

We evaluate the predictive capabilities of our approach %by conducting 
through experiments on both real-world and synthetic datasets. %It is imperative to test on real-world datasets because of their non-stationary features, i.e., different regions have different degrees of variability. It determines whether the learned model can fully adapt and learn these local features, which makes it suitable for solving complex problems. 
For the real-world datasets, we incorporate four tiles, N17E073, N43W080, N45W123, and N47W124 from NASA's Shuttle Radar Topography Mission (SRTM)~\cite{farr2007shuttle}. % as our real-world dataset (from \href{https://dwtkns.com/srtm30m/}{https://dwtkns.com/srtm30m/}). 
The %real-world 
datasets have two-dimensional inputs representing %featuring %dataset that maps features: 
latitude and longitude, with elevation as the output. 
%From a vast amount of elevation maps that encompass the world, we selected four environments: N17E073, N43W080, N45W123, and N47W124, shown in Figure \ref{fig:srtmViz}. N17E073 is characterized by the Western Ghats mountain ranges near the western coast of India. N45W123 represents the Lake Ontario region near Toronto with relatively flat plains. N43W123 points to the dramatic landscape changes in the Pacific Northwest region of the USA, as it encompasses the western coast as well as high volcanic peaks. N47W124 contains the Olympic Peninsula in Washington state. 
As reported in~%Previous work 
\cite{chen2024adaptive}, these datasets exhibit %were assessed %the presence of 
non-stationarity, %features, 
i.e., different regions have varying degrees of variability. %, in these regions. 
This is important for assessing %to determine 
whether the learned model can fully adapt to the local features, thus demonstrating its suitability %making it suitable 
for solving complex problems. For the synthetic dataset, we generate a quantum Gaussian process (QGP) prior using quantum kernels, sample points from it to act as a training dataset, and learn the original QGP. %A two-dimensional quantum Gaussian process, along with local data partitioning for four agents, is depicted in Figure \ref{fig:2DQGPViz}. 
%Finally, we use the three-dimensional Hartmann function, one of the benchmark functions for examining the optimization techniques, as a training dataset generator function. Here, once again, our approach's goal is to learn the function.

% \begin{figure*}[!t]
% 	\includegraphics[width=\textwidth]{figures/srtm_visualization.png}
% 	\centering
% 	\caption{NASA Shuttle Radar Topography Mission's elevation plots for four environments possessing non-stationary features.
%     %used as a real-world dataset for our numerical experiments. All these regions possess non-stationary features, making them viable regions for our model evaluation.
%     }
% 	\label{fig:srtmViz}
% \end{figure*}

% \begin{figure}[!t]
% 	\includegraphics[width=\columnwidth]{figures/2D-QGP.png}
% 	\centering
% 	\caption{(left) 2D Quantum Gaussian Process prior used as a synthetic dataset for our numerical experiments. (right) Partitioning the sampled data points to four mutually exhaustive regions for four agents.}
% 	\label{fig:2DQGPViz}
% \end{figure}

We use two metrics, the NLPD~\eqref{eqn:nlpdtest} and NRMSE~\eqref{eqn:nrmsetest}, which test %different yet 
complementary aspects of %our approach's 
performance. Lower values for both metrics indicate better model performance. Since GPs are %inherently 
probabilistic, NLPD evaluates the quality of the predictive distribution, % by quantifying uncertainty estimates, 
while NRMSE measures the accuracy of the mean predictions. %These metrics are complementary: a learned model with a lower $NLPD$ but higher $NRMSE$ signifies that it is either over-predicting or under-predicting; alternatively, a model with lower $NRMSE$ but higher $NLPD$ indicates that it is incorrectly estimating the uncertainties. Hence, 
%Collectively\
%Together, they provide a comprehensive %picture 
%assessment of the model's capabilities. 
All numerical experiments are %performed 
conducted using classical quantum state vector simulators in Qiskit and PennyLane, employing %. We adopt 
quantum encoding circuits and kernels in %provided by 
\textsc{sQUlearn}~\cite{kreplin2025squlearn}. %Note that this work does not involve a computational complexity analysis comparison because of limited access to real quantum computing hardware. 
%It is to be noted that 
The computational complexity analysis on a quantum simulator run on classical hardware %will 
does not represent %of our methodology's 
the true complexity on current NISQ quantum hardware, hence this work does not involve complexity analysis.

\begin{figure*}[!t]
	\includegraphics[width=\textwidth]{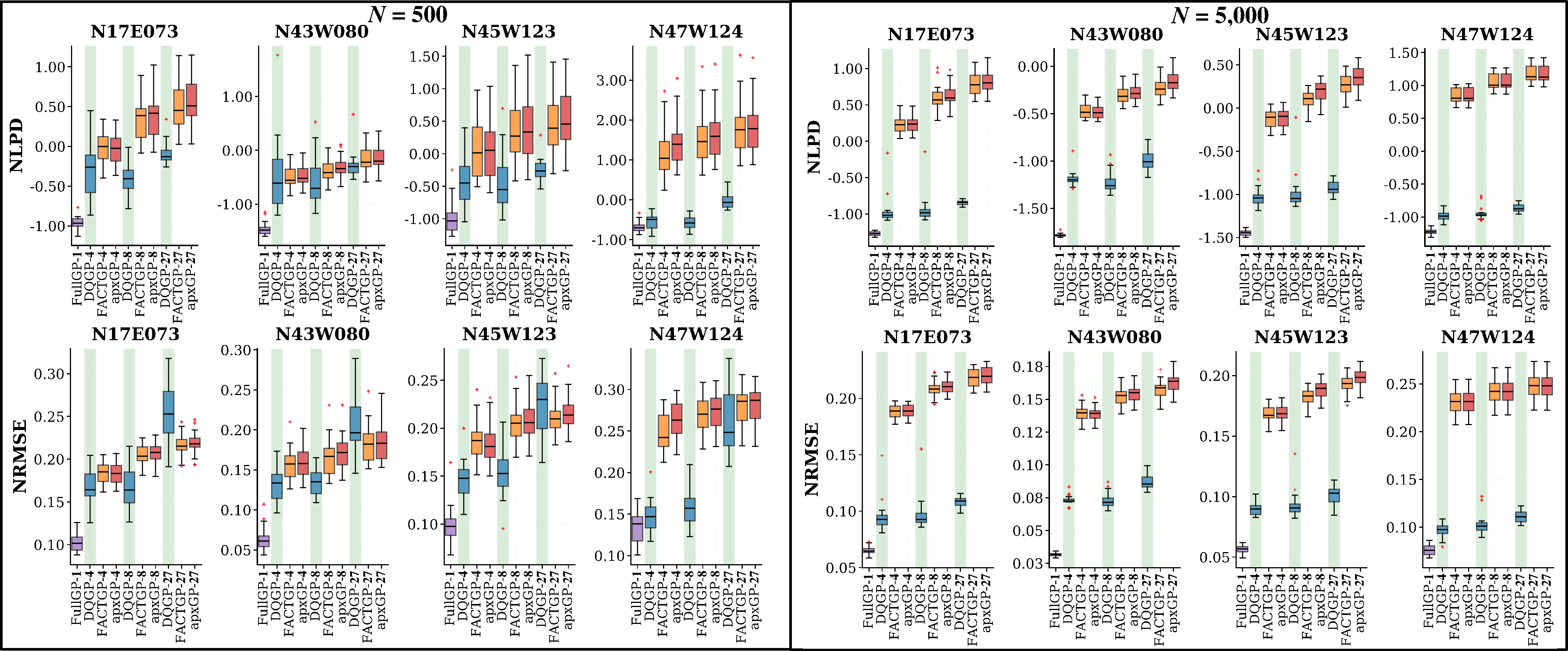}
	\centering
	\caption{Performance of DQGP (green) with the SRTM dataset, compared to Full-GP \cite{williams2006gaussian}, FACT-GP \cite{deisenroth2015distributed}, and apx-GP \cite{xie2019distributed}.}%For N43W080 and N47W124, we have excluded visualizing the apxGP results (worst performance) in NLPD for better readability.}
	\label{fig:srtm}
\end{figure*}

\begin{figure*}[!t]
	\includegraphics[width=\textwidth]{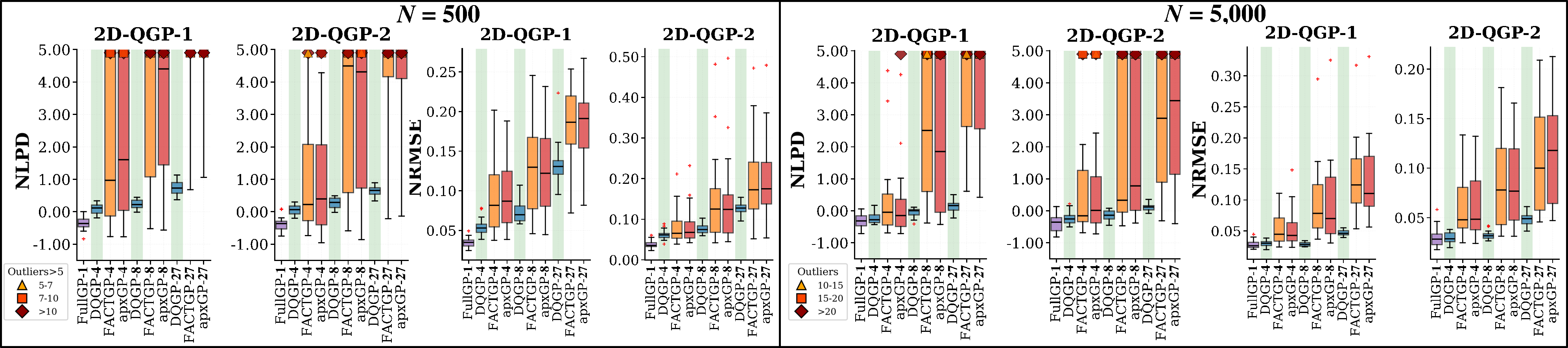}
	\centering
	\caption{Performance of DQGP (green) on 2D QGP prior dataset, compared to Full-GP \cite{williams2006gaussian}, FACT-GP \cite{deisenroth2015distributed}, and apx-GP \cite{xie2019distributed}.}
	\label{fig:2dqgp}
\end{figure*}

% For our experiments, the choice of the number of agents $M = \{1, 4, 8, 27\}$ is based on our regional data partitioning algorithm, as these $M$ allow data to be split into regular grids. For regular grids, the number of fragments in the $D$-dimensional space should be the perfect $D^{\text{th}}$ power, that is, $2^{(D=2)} = 4,\, 2^{(D=3)} = 8, \, 3^{(D=3)} =27 $. As a result of these selections, we generate a well-balanced and spatially coherent agent regions needed for quantum GP's underlying spatial correlation structure.

Across the four environments from the SRTM real-world dataset, we use %work with 
$N=\{500, 5,000\}$ samples, of which $10\%$ is %left 
reserved for testing and the rest $90\%$ is distributed among the agents %to be used 
as their respective training datasets. %Note that here t
The features and target are z-score normalized to % from 
$[-3,3]$ as the quantum parameters---mostly rotational---are naturally bounded. The noise hyperparameter $\sigma^2_\epsilon$ is bounded from $0.1$ to $1$ for all datasets. For quantum encoding, we employ a Hubregtsen Encoding Circuit~\cite{hubregtsen2022training} (Fig.~\ref{fig:approach}) with $q=3$ qubits and $\iota=1$ variational layers, resulting in $P=6$ quantum hyperparameters. We utilize the projected quantum kernel $\kappa_{\textrm{PQK}}$~\eqref{eqn:pqk} with the Pauli $\boldsymbol{\sigma}_{\mathbf{X}}\boldsymbol{\sigma}_{\mathbf{Y}}\boldsymbol{\sigma}_{\mathbf{Z}}$ measurement operator and the Mat\'ern outer kernel%~\eqref{eqn:matern}
, whose parameters are %. Mat\'rn outer kernel's parameter values are 
set to $\ell =1.0$, $\nu = 1.5 $. Quantum hyperparameters are optimized using DR-ADMM (Algorithm~\ref{alg:DR-ADMM}) with shift value %set to 
$\delta =(\pi/8)$ and ADMM parameters: penalty $\rho = 100$ and Lipschitz constant $\mathbi{L} = \{100\}^{M}_{m=1}$. For %experiments with 
the synthetic 2D QGP prior dataset, 
% we use %the number of samples is 
% $N=\{500, 5,000\}$ samples. %Normalization is not required since the data is %Moreover, it does not require normalization, as it is 
% %inherently generated from a QGP. 
% The quantum encoding circuit %we incorporate here 
% is the Hubregtsen Encoding Circuit %depicted in 
% (Fig.~\ref{fig:approach}), with $q=3$ qubits and $\iota=1$ layer resulting in $P=6$ quantum hyperparameters. The quantum kernel %here 
% is %also 
% the projected $\kappa_{\textrm{PQK}}$~\eqref{eqn:pqk} with the Pauli $\boldsymbol{\sigma}_{\mathbf{X}}\boldsymbol{\sigma}_{\mathbf{Y}}\boldsymbol{\sigma}_{\mathbf{Z}}$ measurement operator. %However, t
% The outer kernel is Gaussian %~\eqref{eqn:gaussian}
% with parameter $\gamma = 1.0$. ADMM parameters, 
the setup remains identical to that of the SRTM dataset. 

%\section{Results}

In Table \ref{tab:results_500}, \ref{tab:results_5000}, we present the performance of our method across %various 
datasets of size %and experimental configurations for 
$N=500$ and $N=5,000$, respectively. We report the mean and standard deviation of NLPD$_{\text{test}}$ and NRMSE$_{\text{test}}$ evaluated over 20 %independent %experimental runs 
replications. % to reduce the effect of randomly assigned data. 
% To compare %the performance of 
% our approach with other single- and multi-agent GP methods, we include %provide 
% the evaluation metrics %with respect to 
% for the following methodologies when trained on the same dataset configuration: 
We compare our method with other single- and multi-agent GP methods: % including: 
Full-GP \cite{williams2006gaussian}, FACT-GP \cite{deisenroth2015distributed}, and apxGP~% using the pxADMM consensus algorithm 
\cite{xie2019distributed}. 
% The comparison highlights %assists in determining 
% the predictive performance and uncertainty quantification of our approach relative to existing methods. %On the other hand, the 
% Additionally, comparisons %of values 
% across Tables~\ref{tab:results_500},~\ref{tab:results_5000} allow scalability assessment of network sizes. % can be used to evaluate how our approach and other approaches scale with the size of the network. 
In addition to NLPD and NRMSE, the comparison allows for network scalability assessment. % of network sizes in Table~\ref{tab:results_500},~\ref{tab:results_5000}.
In Fig.~\ref{fig:srtm},~\ref{fig:2dqgp}, we present the performance in the SRTM and 2D QGP datasets, respectively.
% We also present boxplots to illustrate %for 
% the performance %of all approaches with 
% in the SRTM and 2D QGP datasets in Fig.~\ref{fig:srtm},~\ref{fig:2dqgp} respectively.

Aggregating the improvements across the four SRTM environments for $N=500$ and $N=5,000$ datasets, DQGP achieves $32.4\% \pm 22.7\%$ lower $\text{NRMSE}_{\text{test}}$ than FACT-GP, and $32.9\% \pm 22.7\%$ lower $\text{NRMSE}_{\text{test}}$ than apxGP. %Although DQGP follows the standard distributed computing trade-off---exhibiting an average $\text{NRMSE}_{\text{test}}$ that is $83.1\%$ higher than the centralized Full-GP---it significantly outperforms other distributed baselines. 
Moreover, the results suggest that DQGP scales effectively with $M$, i.e., while all distributed methods see a rise in error as $M$ increases from 4 to 27, DQGP retains its predictive edge more robustly. For all approaches, a trade-off is observed in $\text{NLPD}_{\text{test}}$ as the number of agents increases. Specifically, DQGP achieves a $349.8\% \pm 369.0\%$ improvement in $\text{NLPD}_{\text{test}}$ relative to FACT-GP and a $435.5\% \pm 810.2\%$ improvement relative to apxGP, with the highest gains appearing in $M=27$ case.

For the synthetic dataset, DQGP achieves on average $47.0\% \pm 14.0\%$ lower $\text{NRMSE}_{\text{test}}$ than FACT-GP and $47.0\% \pm 14.2\%$ lower $\text{NRMSE}_{\text{test}}$ than apxGP. %Compared to the single-agent Full-GP, DQGP shows a performance gap of $69.4\%$ in aggregated $\text{NRMSE}_{\text{test}}$, which is a narrower margin than that observed in the SRTM datasets. 
With regards to uncertainty quantification, DQGP achieves a $105.5\% \pm 19.7\%$ improvement in $\text{NLPD}_{\text{test}}$ relative to FACT-GP and a $101.7\% \pm 7.9\%$ improvement over apxGP. The enhanced DQGP performance relative to other distributed approaches is even more pronounced in 2D Quantum GP prior dataset. %For both real and synthetic datasets, the single agent Full-GP method acts as the performance benchmark for distributed approaches.

Overall, the proposed DQGP %\textit{Distributed Quantum Gaussian Process} with \textit{Distributed consensus Riemannian ADMM}
demonstrates a substantial advantage in prediction %accuracy 
over classical distributed counterparts, as is evident from the lower NRMSE$_{\text{test}}$ values for both $N=\{500, 5,000\}$. %Upon comparing $NLPD_{\text{test}}$ between $N=500$ and $N=5000$, we can conclude that as the network size increases, our method becomes overconfident; however, it excels in data-scarce regimes.
It %succeeds 
also successfully quantifies %in quantifying 
the uncertainty for both $N=\{500, 5,000\}$, as reflected by lower NLPD$_\text{test}$ values. For some datasets with $N=500$ and $M=27$, FACT-GP yields lower NRMSE$_\text{test}$ than DQGP. The rare performance gap stems from the high-dimensional mapping of sparse data using expressive quantum kernels. When a small amount of classical data is embedded into the quantum Hilbert space, the corresponding quantum states become nearly orthogonal. This eliminates correlations and triggers barren plateau~\cite{larocca2025barren, mcclean2018barren}.
% due to its block-diagonal posterior covariance approximation, which %leads to a 
% produces less conservative and more stable uncertainty estimates. In contrast, DQGP with DR-ADMM %focuses on 
% prioritizes on aligning the global objective across the agents, resulting in more accurate mean predictions and only occasionally resulting in %and overly 
% conservative uncertainty estimates.

%%%%%%%%%%%%%%%%%%%%%%%%%%%%%%%%%%%%%%%%%%%%%%%%%%%%%%%%%%%%%%%%%%%%%%%%

\section{Conclusion}
This paper %proposes 
introduces a distributed quantum gaussian process (DQGP) method that scales the expressive power  %allowing the enhanced expressivity 
of quantum kernels to %be applied to expansive 
real-world, non-stationary datasets. % through distributed learning. %datasets by increasing their scalability. 
%Moreover, to deal with 
To address the non-Euclidean quantum hyperparameter optimization, we propose a Distributed consensus Riemannian ADMM (DR-ADMM) approach. % implemented in a distributed setting. 
Experiments on real-world and synthetic datasets demonstrate %To evaluate our proposed approach against competitive single-agent and multi-agent gaussian process regression methods, we conducted experiments with them on real-world and synthetic datasets. Our experiments reveal a 
enhanced %predictive accuracy % performance with a slightly 
%and competitive uncertainty estimation 
performance compared to classical distributed methods. %relative to other methods. 
The results highlight the potential of DQGP for scalable %promising results of our approach highlight a path forward for large-scale, 
probabilistic modeling on hybrid classical-quantum systems. %Ongoing work focuses on active learning tasks, paving the way for quantum-enhanced autonomous systems. %more accurate and expressive distributed QGP models.

%The promising results of our approach have encouraged several future directions. We want to test our approach with higher-dimensional data, as that is where quantum computing shines. We want to explore ways to make our approach decentralized. Lastly, we want to present a computational complexity analysis of our approach by implementing it on real quantum hardware.

%%%%%%%%%%%%%%%%%%%%%%%%%%%%%%%%%%%%%%%%%%%%%%%%%%%%%%%%%%%%%%%%%%%%%%%%

%%% The acknowledgments section is defined using the "acks" environment
%%% (rather than an unnumbered section). The use of this environment 
%%% ensures the proper identification of the section in the article 
%%% metadata as well as the consistent spelling of the heading.

% \begin{acks}
% If you wish to include any acknowledgments in your paper (e.g., to 
% people or funding agencies), please do so using the `\texttt{acks}' 
% environment. Note that the text of your acknowledgments will be omitted
% if you compile your document with the `\texttt{anonymous}' option.
% \end{acks}

%%%%%%%%%%%%%%%%%%%%%%%%%%%%%%%%%%%%%%%%%%%%%%%%%%%%%%%%%%%%%%%%%%%%%%%%

%%% The next two lines define, first, the bibliography style to be 
%%% applied, and, second, the bibliography file to be used.
% \balance 
\bibliographystyle{ACM-Reference-Format} 
\bibliography{sample}

@article{kontoudis2025multi,
  title={Multi-Agent Federated Learning Using Covariance-Based Nearest Neighbor {Gaussian} Processes},
  author={Kontoudis, George P and Stilwell, Daniel J},
  journal={IEEE Transactions on Machine Learning in Communications and Networking},
  volume={4},
  pages={115--138},
  year={2025},
  publisher={IEEE}
}

@inproceedings{kontoudis2021decentralized,
  title={Decentralized nested {G}aussian processes for multi-robot systems},
  booktitle = {IEEE International Conference on Robotics and Automation},
  author={Kontoudis, George P and Stilwell, Daniel J},
  pages = {8881--8887},
  year = {2021}
}

@inproceedings{kontoudis2023decentralized,
  title={Decentralized federated learning using {G}aussian processes},
  author={Kontoudis, George P and Stilwell, Daniel J},
  booktitle={IEEE International Symposium on Multi-Robot and Multi-Agent Systems},
  pages={1--7},
  year={2023}
}

@article{chang2014multi,
  title={Multi-agent distributed optimization via inexact consensus {ADMM}},
  author={Chang, Tsung-Hui and Hong, Mingyi and Wang, Xiangfeng},
  journal={IEEE Transactions on Signal Processing},
  volume={63},
  number={2},
  pages={482--497},
  year={2014},
  publisher={IEEE}
}

@book{boyd2011admm,
year = {2011},
volume = {3},
journal = {Foundations and Trends in Machine Learning},
title = {Distributed Optimization and Statistical Learning via the Alternating Direction Method of Multipliers},
number = {1},
author = {Stephen Boyd and Neal Parikh and Eric Chu and Borja Peleato and Jonathan Eckstein}
}

@article{liu2020gaussian,
  title={When {G}aussian process meets big data: A review of scalable {GP}s},
  author={Liu, Haitao and Ong, Yew-Soon and Shen, Xiaobo and Cai, Jianfei},
  journal={IEEE Transactions on Neural Networks and Learning Systems},
  year={2020},
  volume={31},
  number={11},
  pages={4405--4423},
  publisher={IEEE}
}

@inproceedings{liu2018generalized,
  title={Generalized Robust {B}ayesian Committee Machine for Large-scale {G}aussian Process Regression},
  author={Liu, Haitao and Cai, Jianfei and Wang, Yi and Ong, Yew Soon},
  booktitle={International Conference on Machine Learning},
  pages={3131--3140},
  year={2018}
}

@article{wierichs2022general,
  title={General parameter-shift rules for {Q}uantum gradients},
  author={Wierichs, David and Izaac, Josh and Wang, Cody and Lin, Cedric Yen-Yu},
  journal={Quantum},
  volume={6},
  pages={677},
  year={2022},
  publisher={Verein zur F{\"o}rderung des Open Access Publizierens in den Quantenwissenschaften}
}

@inproceedings{parekh2021quantum,
  title={Quantum algorithms and simulation for parallel and distributed quantum computing},
  author={Parekh, Rhea and Ricciardi, Andrea and Darwish, Ahmed and DiAdamo, Stephen},
  booktitle={IEEE/ACM Intern. Workshop on Quantum Computing Software},
  pages={9--19},
  year={2021}
}

@article{kontoudis2024scalable,
  title={Scalable, federated Gaussian process training for decentralized multi-agent systems},
  author={Kontoudis, George P and Stilwell, Daniel J},
  journal={IEEE Access},
  volume={12},
  pages={77800--77815},
  year={2024},
  publisher={IEEE}
}

@article{chen2022quantum,
  title={Quantum algorithm for Gaussian process regression},
  author={Chen, Meng-Han and Yu, Chao-Hua and Gao, Jian-Liang and Yu, Kai and Lin, Song and Guo, Gong-De and Li, Jing},
  journal={Physical Review A},
  volume={106},
  number={1},
  pages={012406},
  year={2022},
  publisher={APS}
}

@article{zhao2019quantum,
  title={Quantum-assisted Gaussian process regression},
  author={Zhao, Zhikuan and Fitzsimons, Jack K and Fitzsimons, Joseph F},
  journal={Physical Review A},
  volume={99},
  number={5},
  pages={052331},
  year={2019},
  publisher={APS}
}

@article{farooq2024quantum,
  title={Quantum-assisted Hilbert-space Gaussian process regression},
  author={Farooq, Ahmad and Galvis-Florez, Cristian A and S{\"a}rkk{\"a}, Simo},
  journal={Physical Review A},
  volume={109},
  number={5},
  pages={052410},
  year={2024},
  publisher={APS}
}

@article{kus2021sparse,
  title={Sparse quantum Gaussian processes to counter the curse of dimensionality},
  author={Ku{\'s}, Gawe{\l} I and van der Zwaag, Sybrand and Bessa, Miguel A},
  journal={Quantum Machine Intelligence},
  volume={3},
  number={1},
  pages={6},
  year={2021},
  publisher={Springer}
}

@article{hubregtsen2022training,
  title={Training quantum embedding kernels on near-term quantum computers},
  author={Hubregtsen, Thomas and Wierichs, David and Gil-Fuster, Elies and Derks, Peter-Jan HS and Faehrmann, Paul K and Meyer, Johannes Jakob},
  journal={Physical Review A},
  volume={106},
  number={4},
  pages={042431},
  year={2022},
  publisher={APS}
}

@article{huang2021power,
  title={Power of data in quantum machine learning},
  author={Huang, Hsin-Yuan and Broughton, Michael and Mohseni, Masoud and Babbush, Ryan and Boixo, Sergio and Neven, Hartmut and McClean, Jarrod R},
  journal={Nature communications},
  volume={12},
  number={1},
  pages={2631},
  year={2021},
  publisher={Nature Publishing Group UK London}
}

@article{arceci2024gaussian,
  title={Gaussian process model kernels for noisy optimization in variational quantum algorithms},
  author={Arceci, Luca and Kuzmin, Viacheslav and Van Bijnen, Rick},
  journal={arXiv preprint arXiv:2412.13271},
  year={2024}
}

@article{smith2023faster,
  title={Faster variational quantum algorithms with quantum kernel-based surrogate models},
  author={Smith, Alistair WR and Paige, AJ and Kim, MS},
  journal={Quantum Science and Technology},
  volume={8},
  number={4},
  pages={045016},
  year={2023},
  publisher={IOP Publishing}
}

@article{xie2019distributed,
  title={Distributed Gaussian processes hyperparameter optimization for big data using proximal ADMM},
  author={Xie, Ang and Yin, Feng and Xu, Yue and Ai, Bo and Chen, Tianshi and Cui, Shuguang},
  journal={IEEE Signal Processing Letters},
  volume={26},
  number={8},
  pages={1197--1201},
  year={2019},
  publisher={IEEE}
}

@book{williams2006gaussian,
  title={Gaussian processes for machine learning},
  author={Williams, Christopher KI and Rasmussen, Carl Edward},
  volume={2},
  number={3},
  year={2006},
  publisher={MIT press Cambridge, MA}
}

@inproceedings{deisenroth2015distributed,
  title={Distributed gaussian processes},
  author={Deisenroth, Marc and Ng, Jun Wei},
  booktitle={International conference on machine learning},
  pages={1481--1490},
  year={2015},
  organization={PMLR}
}

@article{galvis2025provable,
  title={Provable Quantum Algorithm Advantage for Gaussian Process Quadrature},
  author={Galvis-Florez, Cristian A and Farooq, Ahmad and S{\"a}rkk{\"a}, Simo},
  journal={arXiv preprint arXiv:2502.14467},
  year={2025}
}

@article{larocca2025barren,
  title={Barren plateaus in variational quantum computing},
  author={Larocca, Martin and Thanasilp, Supanut and Wang, Samson and Sharma, Kunal and Biamonte, Jacob and Coles, Patrick J and Cincio, Lukasz and McClean, Jarrod R and Holmes, Zo{\"e} and Cerezo, Marco},
  journal={Nature Reviews Physics},
  pages={1--16},
  year={2025},
  publisher={Nature Publishing Group UK London}
}

@article{cerezo2021variational,
  title={Variational quantum algorithms},
  author={Cerezo, Marco and Arrasmith, Andrew and Babbush, Ryan and Benjamin, Simon C and Endo, Suguru and Fujii, Keisuke and McClean, Jarrod R and Mitarai, Kosuke and Yuan, Xiao and Cincio, Lukasz and others},
  journal={Nature Reviews Physics},
  volume={3},
  number={9},
  pages={625--644},
  year={2021},
  publisher={Nature Publishing Group UK London}
}

@inproceedings{lloyd2010quantum,
  title={Quantum algorithm for solving linear systems of equations},
  author={Lloyd, Seth},
  booktitle={APS March Meeting Abstracts},
  volume={2010},
  pages={D4--002},
  year={2010}
}

@article{guerreschi2019qaoa,
  title={QAOA for Max-Cut requires hundreds of qubits for quantum speed-up},
  author={Guerreschi, Gian Giacomo and Matsuura, Anne Y},
  journal={Scientific reports},
  volume={9},
  number={1},
  pages={6903},
  year={2019},
  publisher={Nature Publishing Group UK London}
}

@article{kreplin2025squlearn,
  title={sQUlearn: a Python library for quantum machine learning},
  author={Kreplin, David A and Willmann, Moritz and Schnabel, Jan and Rapp, Frederic and Hagel{\"u}ken, Manuel and Roth, Marco},
  journal={IEEE Software},
  year={2025},
  publisher={IEEE}
}

@article{rapp2024quantum,
  title={Quantum gaussian process regression for bayesian optimization},
  author={Rapp, Frederic and Roth, Marco},
  journal={Quantum Machine Intelligence},
  volume={6},
  number={1},
  pages={5},
  year={2024},
  publisher={Springer}
}

@article{kanagawa2018gaussian,
  title={Gaussian processes and kernel methods: A review on connections and equivalences},
  author={Kanagawa, Motonobu and Hennig, Philipp and Sejdinovic, Dino and Sriperumbudur, Bharath K},
  journal={arXiv preprint arXiv:1807.02582},
  year={2018}
}

@article{preskill2018quantum,
  title={Quantum computing in the NISQ era and beyond},
  author={Preskill, John},
  journal={Quantum},
  volume={2},
  pages={79},
  year={2018},
  publisher={Verein zur F{\"o}rderung des Open Access Publizierens in den Quantenwissenschaften}
}

@article{benedetti2019parameterized,
  title={Parameterized quantum circuits as machine learning models},
  author={Benedetti, Marcello and Lloyd, Erika and Sack, Stefan and Fiorentini, Mattia},
  journal={Quantum science and technology},
  volume={4},
  number={4},
  pages={043001},
  year={2019},
  publisher={IOP Publishing}
}

@article{ghahramani2013bayesian,
  title={Bayesian non-parametrics and the probabilistic approach to modelling},
  author={Ghahramani, Zoubin},
  journal={Philosophical Transactions of the Royal Society A: Mathematical, Physical and Engineering Sciences},
  volume={371},
  number={1984},
  pages={20110553},
  year={2013},
  publisher={The Royal Society Publishing}
}

@book{shi2011gaussian,
  title={Gaussian process regression analysis for functional data},
  author={Shi, Jian Qing and Choi, Taeryon},
  year={2011},
  publisher={CRC press}
}

@article{schuld2019quantum,
  title={Quantum machine learning in feature Hilbert spaces},
  author={Schuld, Maria and Killoran, Nathan},
  journal={Physical review letters},
  volume={122},
  number={4},
  pages={040504},
  year={2019},
  publisher={APS}
}

@incollection{schuld2021quantum,
  title={Quantum models as kernel methods},
  author={Schuld, Maria and Petruccione, Francesco},
  booktitle={Machine Learning with Quantum Computers},
  pages={217--245},
  year={2021},
  publisher={Springer}
}

@book{lee2018introduction,
  title={Introduction to Riemannian manifolds},
  author={Lee, John M},
  volume={2},
  year={2018},
  publisher={Springer}
}

@article{li2022riemannian,
  title={A riemannian admm},
  author={Li, Jiaxiang and Ma, Shiqian and Srivastava, Tejes},
  journal={arXiv preprint arXiv:2211.02163},
  year={2022}
}

@article{zanardi2006ground,
  title={Ground state overlap and quantum phase transitions},
  author={Zanardi, Paolo and Paunkovi{\'c}, Nikola},
  journal={Physical Review E—Statistical, Nonlinear, and Soft Matter Physics},
  volume={74},
  number={3},
  pages={031123},
  year={2006},
  publisher={APS}
}

@article{gil2024expressivity,
  title={On the expressivity of embedding quantum kernels},
  author={Gil-Fuster, Elies and Eisert, Jens and Dunjko, Vedran},
  journal={Machine Learning: Science and Technology},
  volume={5},
  number={2},
  pages={025003},
  year={2024},
  publisher={IOP Publishing}
}

@article{heinosaari2008notes,
  title={Notes on joint measurability of quantum observables},
  author={Heinosaari, Teiko and Reitzner, Daniel and Stano, Peter},
  journal={Foundations of Physics},
  volume={38},
  number={12},
  pages={1133--1147},
  year={2008},
  publisher={Springer}
}

@article{manzhos2024degeneration,
  title={Degeneration of kernel regression with Matern kernels into low-order polynomial regression in high dimension},
  author={Manzhos, Sergei and Ihara, Manabu},
  journal={The Journal of Chemical Physics},
  volume={160},
  number={2},
  year={2024},
  publisher={AIP Publishing}
}

@article{farr2007shuttle,
  title={The shuttle radar topography mission},
  author={Farr, Tom G and Rosen, Paul A and Caro, Edward and Crippen, Robert and Duren, Riley and Hensley, Scott and Kobrick, Michael and Paller, Mimi and Rodriguez, Ernesto and Roth, Ladislav and others},
  journal={Reviews of geophysics},
  volume={45},
  number={2},
  year={2007},
  publisher={Wiley Online Library}
}

@article{chen2024adaptive,
  title={Adaptive robotic information gathering via non-stationary Gaussian processes},
  author={Chen, Weizhe and Khardon, Roni and Liu, Lantao},
  journal={The International Journal of Robotics Research},
  volume={43},
  number={4},
  pages={405--436},
  year={2024},
  publisher={SAGE Publications Sage UK: London, England}
}

@article{mcclean2018barren,
  title={Barren plateaus in quantum neural network training landscapes},
  author={McClean, Jarrod R and Boixo, Sergio and Smelyanskiy, Vadim N and Babbush, Ryan and Neven, Hartmut},
  journal={Nature communications},
  volume={9},
  number={1},
  pages={4812},
  year={2018},
  publisher={Nature Publishing Group UK London}
}

@article{patti2021entanglement,
  title={Entanglement devised barren plateau mitigation},
  author={Patti, Taylor L and Najafi, Khadijeh and Gao, Xun and Yelin, Susanne F},
  journal={Physical Review Research},
  volume={3},
  number={3},
  pages={033090},
  year={2021},
  publisher={APS}
}

@article{sack2022avoiding,
  title={Avoiding barren plateaus using classical shadows},
  author={Sack, Stefan H and Medina, Raimel A and Michailidis, Alexios A and Kueng, Richard and Serbyn, Maksym},
  journal={PRX Quantum},
  volume={3},
  number={2},
  pages={020365},
  year={2022},
  publisher={APS}
}

@article{peng2025breaking,
  title={Breaking Through Barren Plateaus: Reinforcement Learning Initializations for Deep Variational Quantum Circuits},
  author={Peng, Yifeng and Li, Xinyi and Zhang, Zhemin and Chen, Samuel Yen-Chi and Liang, Zhiding and Wang, Ying},
  journal={arXiv preprint arXiv:2508.18514},
  year={2025}
}

@article{cunningham2025investigating,
  title={Investigating and mitigating barren plateaus in variational quantum circuits: a survey},
  author={Cunningham, Jack and Zhuang, Jun},
  journal={Quantum Information Processing},
  volume={24},
  number={2},
  pages={48},
  year={2025},
  publisher={Springer}
}

@article{arute2019quantum,
  title={Quantum supremacy using a programmable superconducting processor},
  author={Arute, Frank and Arya, Kunal and Babbush, Ryan and Bacon, Dave and Bardin, Joseph C and Barends, Rami and Biswas, Rupak and Boixo, Sergio and Brandao, Fernando GSL and Buell, David A and others},
  journal={Nature},
  volume={574},
  number={7779},
  pages={505--510},
  year={2019},
  publisher={Nature Publishing Group UK London}
}

@inproceedings{glendinning2005bloch,
  title={The bloch sphere},
  author={Glendinning, Ian},
  booktitle={QIA meeting. Vienna},
  year={2005}
}

@article{lim2012matrix,
  title={Matrix power means and the Karcher mean},
  author={Lim, Yongdo and P{\'a}lfia, Mikl{\'o}s},
  journal={Journal of Functional Analysis},
  volume={262},
  number={4},
  pages={1498--1514},
  year={2012},
  publisher={Elsevier}
}

@article{park2023quantum,
  title={Quantum multi-agent reinforcement learning for autonomous mobility cooperation},
  author={Park, Soohyun and Kim, Jae Pyoung and Park, Chanyoung and Jung, Soyi and Kim, Joongheon},
  journal={IEEE Communications Magazine},
  volume={62},
  number={6},
  pages={106--112},
  year={2023},
  publisher={IEEE}
}

@inproceedings{chen2025quantum,
  title={Quantum-train-based distributed multi-agent reinforcement learning},
  author={Chen, Kuan-Cheng and Chen, Samuel Yen-Chi and Liu, Chen-Yu and Leung, Kin K},
  booktitle={2025 IEEE Symposium for Multidisciplinary Computational Intelligence Incubators (MCII Companion)},
  pages={1--5},
  year={2025},
  organization={IEEE}
}

@Inbook{Williams2011,
author="Williams, Colin P.",
title="Quantum Gates",
bookTitle="Explorations in Quantum Computing",
year="2011",
publisher="Springer London",
address="London",
pages="51--122"
}

%%%%%%%%%%%%%%%%%%%%%%%%%%%%%%%%%%%%%%%%%%%%%%%%%%%%%%%%%%%%%%%%%%%%%%%%

\end{document}